\documentclass[twocolumn,10pt,superscriptaddress,amsmath,amssymb,aps,pre,floatfix,longbibliography]{revtex4-2}

\usepackage{graphicx}% Include figure files
\usepackage{bm}% bold math
\usepackage{times}
\usepackage{comment}
\usepackage{diagbox}
\usepackage{amsmath,amssymb,amsthm}
\usepackage{tabularx}
\usepackage{array}
\usepackage{makecell}
\usepackage{tikz}
\usepackage{multirow}
\usepackage{xcolor}
\usepackage{enumerate}
\usepackage[colorlinks=true,linkcolor=blue,citecolor=red, linktocpage=true,breaklinks=true]{hyperref}

\newtheorem{theorem}{Theorem}
\newtheorem{lemma}[theorem]{Lemma}

\makeatletter
\newcommand{\Rmnum}[1]{\expandafter\@slowromancap\romannumeral #1@}
\makeatother

\renewcommand{\arraystretch}{1.5}

\begin{document}

\title{Minimal nonintegrable models with three-site interactions}

\author{Wen-Ming Fan}
\affiliation{School of Physics, Northwest University, Xi'an 710127, China}

\author{Kun \surname{Hao}}
\affiliation{Institute of Modern Physics, Northwest University, Xi'an 710127, China}
\affiliation{Shaanxi Key Laboratory for Theoretical Physics Frontiers, Xi'an 710127, China}
\affiliation{Peng Huanwu Center for Fundamental Theory, Xi'an 710127, China}
\affiliation{Fundamental Discipline Research Center for Quantum Science and Technology of Shaanxi Province, Xi'an 710127, China}

\author{Yang-Yang Chen}
\affiliation{Institute of Modern Physics, Northwest University, Xi'an 710127, China}
\affiliation{Shaanxi Key Laboratory for Theoretical Physics Frontiers, Xi'an 710127, China}
\affiliation{Peng Huanwu Center for Fundamental Theory, Xi'an 710127, China}
\affiliation{Fundamental Discipline Research Center for Quantum Science and Technology of Shaanxi Province, Xi'an 710127, China}

\author{Xiao-Hui Wang}
%\email{xhwang@nwu.edu.cn}
\affiliation{School of Physics, Northwest University, Xi'an 710127, China}
\affiliation{Shaanxi Key Laboratory for Theoretical Physics Frontiers, Xi'an 710127, China}
\affiliation{Peng Huanwu Center for Fundamental Theory, Xi'an 710127, China}
\affiliation{Fundamental Discipline Research Center for Quantum Science and Technology of Shaanxi Province, Xi'an 710127, China}

\author{Kun Zhang}
\email{kunzhang@nwu.edu.cn}
\affiliation{School of Physics, Northwest University, Xi'an 710127, China}
\affiliation{Shaanxi Key Laboratory for Theoretical Physics Frontiers, Xi'an 710127, China}
\affiliation{Peng Huanwu Center for Fundamental Theory, Xi'an 710127, China}
\affiliation{Fundamental Discipline Research Center for Quantum Science and Technology of Shaanxi Province, Xi'an 710127, China}

\author{Vladimir \surname{Korepin}}
\affiliation{C.N. Yang Institute for Theoretical Physics, Stony Brook University, New York 11794, USA}

\begin{abstract}
We study integrability breaking in translationally invariant spin-$1/2$ chains with genuine three-site interactions. Using a two-qubit composite representation, we prove that the deformed Fredkin spin chain is nonintegrable for any nonzero deformation parameter, although its Hamiltonian decomposes into coupled integrable building blocks. We then extract the minimal injective models, defined as the simplest models needed for a rigorous nonintegrability test, from the Hamiltonian density of the deformed Fredkin spin chain. Among the sixteen such models, six satisfy the Reshetikhin condition and are integrable, while two satisfy Hokkyo's criterion and are nonintegrable away from special coefficient values. Returning to one-site translation-invariant spin-$1/2$ chains leaves four genuine three-site models: two integrable models and two generically nonintegrable models. These results give a compact classification of minimal three-site models obtained from the deformed Fredkin spin chain and identify a local mechanism by which coupling individually integrable structures can destroy integrability beyond nearest-neighbor interactions.
\end{abstract}

\maketitle

\section{Introduction}\label{sec:intro}

Integrable quantum many-body systems occupy a distinguished position in modern theoretical physics because their dynamics is strongly constrained by extensive families of local conserved charges and, in many cases, is exactly solvable \cite{takahashi1999thermodynamics,korepin1997quantum,baxter2016exactly}. They provide benchmarks for quantum phase transitions and critical phenomena \cite{lieb1968absence,sachdev1999quantum,kitanine2011form}. 
Their conserved charges also prevent conventional thermalization and lead to distinctive nonequilibrium behavior \cite{rigol2007relaxation,rigol2008thermalization,Pozsgay2013GeneralizedGibbsHeisenberg,langen2015experimental,essler2016quench,vidmar2016generalized}. Beyond spin chains, integrable structures appear in quantum circuits \cite{vanicat2018integrable,Aleiner2021,miao2024floquet,paletta2025integrabilitygeometry}, open quantum dynamics \cite{2016Exact,2019Dissipative,2020Yang,deLeeuw2021ConstructingIL}, and high-energy physics \cite{zamolodchikov1989integrable,fock2000duality,minahan2003bethe,staudacher2005factorized,beisert2012review,cavaglia2016mathrm}. 
Thus integrability is not a generic property of interacting quantum matter, but an exceptional organizing principle for many-body systems.

Understanding how this structure breaks down is therefore as important as constructing integrable models. 
Nonintegrable quantum many-body systems are expected to describe generic interacting dynamics, yet analytical criteria for proving nonintegrability have only recently become systematic. 
Physically, nonintegrable models differ from integrable ones through chaotic dynamics and thermalization consistent with the eigenstate thermalization hypothesis (ETH) \cite{deutsch1991quantum,srednicki1994chaos,d2016quantum}. This distinction connects information scrambling and thermalization \cite{srdinvsek2021signatures,nunez2018integrability} with quantum-information applications such as benchmarking quantum simulators and quantum computations \cite{maruyoshi2023conserved,yao2021adaptive}.

%At the physical level, nonintegrable models display qualitatively different behavior from integrable ones, including chaotic dynamics and thermalization consistent with the eigenstate thermalization hypothesis (ETH) \cite{deutsch1991quantum,srednicki1994chaos,d2016quantum}. These properties underlie a wide range of phenomena, from information scrambling and thermalization in quantum field theories \cite{srdinvsek2021signatures,nunez2018integrability}, black-hole physics \cite{shenker2014black}, and string-theoretic models \cite{giataganas2014non,de2021free,nunez2018non}, to practical roles in quantum information science such as benchmarking quantum simulators and quantum computations \cite{maruyoshi2023conserved,yao2021adaptive}.

For one-dimensional spin chains with nearest-neighbor interactions, integrability and nonintegrability are now comparatively well understood. On the integrable side, the Bethe ansatz and the quantum inverse scattering method provide powerful constructive frameworks \cite{Faddeev1996How,korepin1997quantum}. 
On the nonintegrable side, rigorous methods initiated by Shiraishi \cite{shiraishi2019proof} and based on the absence of local conserved quantities have led to complete classifications for broad classes of nearest-neighbor spin-$1/2$ chains \cite{yamaguchi2024classification,yamaguchi2024proof}, and have since been applied to many other spin models \cite{chiba2024proofIsing,chiba2025proofhighdimension,futami2025absence,shiraishi2024s,Futami2025AbsenceON,shiraishi2024absenceXIX,shiraishi2025complete,Ishii2026ProofHolstein,hokkyo2024proof,park2025proofspin1,Shiraishi2025DichotomyTS}. 

By contrast, systems with longer-range interactions remain much less systematically explored. 
Among them, spin chains with genuine three-site interactions constitute the minimal extension beyond the nearest-neighbor paradigm. Although several integrable three-site models are known, including free-fermions-in-disguise (FFD) models \cite{fendley2019free_in}, the Bariev model \cite{bariev1991integrable}, and the folded XXZ chain \cite{zadnik2021folded1,zadnik2021folded2}, prominent nonintegrable examples such as the PXP model \cite{park2025graph,park2025nonintegrability}, extended Ising models \cite{udupa2023weak}, and the Fredkin spin chain \cite{fan2025absence} show that a general structural understanding is still lacking.

The central question addressed in this work is therefore:
what is the minimal local structure required to destroy integrability in translationally invariant spin-$1/2$ chains with genuine three-site interactions?
We focus on Hamiltonians built from Pauli-string three-site terms and call a model minimally nonintegrable if it satisfies three requirements. 
First, it contains genuine three-site interactions. 
Second, it saturates the minimal number of terms required by injectivity, a structural assumption in Hokkyo's nonintegrability criterion \cite{hokkyo2025rigorous}. Third, it is nonintegrable, whereas removing either three-site term leaves an integrable model. In the examples studied here, the reduced models are FFD integrable. This notion is motivated by recent rigorous results showing that injectivity and low-order conservation laws strongly constrain the boundary between integrability and nonintegrability \cite{hokkyo2025rigorous,hokkyo2025integrability}. 

%Our analysis builds on recent progress in rigorous criteria for integrability and nonintegrability. For nearest-neighbor Hamiltonians, Shiraishi established a general strategy for proving the absence of local conserved charges \cite{shiraishi2019proof}, although its direct implementation can be computationally involved. This approach and its variants have since been applied to many models \cite{chiba2024proofIsing,shiraishi2024absenceXIX,shiraishi2025complete,park2025graph,park2025nonintegrability,hokkyo2024proof,park2025proofspin1,Shiraishi2025DichotomyTS,futami2025absence,shiraishi2024s,chiba2025proofhighdimension,Futami2025AbsenceON,fan2025absence}. A substantial simplification was recently obtained by Hokkyo \cite{hokkyo2025rigorous}: for nearest-neighbor Hamiltonians satisfying injectivity and the $2$-local conservation condition, nonintegrability can be certified by ruling out an admissible $3$-local conserved charge. Conversely, the Reshetikhin or boost-generated $3$-local conservation law implies an infinite hierarchy of commuting local charges, and hence integrability \cite{hokkyo2025integrability}. Although these criteria are formulated for nearest-neighbor interactions, they can be applied to three-site spin-$1/2$ Hamiltonians after grouping two physical spins into one composite site of local dimension four. 

A natural starting point for this program is the deformed Fredkin spin chain \cite{salberger2017deformed}. 
This model generalizes the original Fredkin chain \cite{salberger2018fredkin}, whose ground state is an exactly known, highly entangled superposition of Dyck paths, by introducing a deformation parameter that changes the entanglement structure \cite{salberger2018fredkin,zhang2017entropy}. 
Here we focus not on its ground-state physics, but on the structure of its Hamiltonian. 
In Pauli-matrix form, the Hamiltonian decomposes into several individually integrable building blocks, including XXZ-type terms, Ising-type interactions, and multiple FFD models \cite{fendley2019free_in}. 
The deformed Fredkin chain therefore provides a controlled setting for identifying irreducible mechanisms of integrability breaking. 
Although the nonintegrability of the original Fredkin chain has been rigorously established \cite{fan2025absence}, the corresponding question for the deformed Fredkin chain has remained open.

In this work, we use the deformed Fredkin Hamiltonian as a structured source of minimal three-site models and combine charge-based criteria with the frustration-graph diagnostic for FFD models \cite{elman2021free_behind}. We first prove that the spin-$1/2$ deformed Fredkin chain with periodic boundary conditions is nonintegrable for any nonzero deformation parameter. We then classify the minimal injective models in its composite two-site density: among sixteen models, six satisfy the Reshetikhin condition, two are generically nonintegrable by Hokkyo's criterion, and the remaining eight are not resolved by the direct two-local test; see Table~\ref{table:min_H_classfication}. Imposing the unblocking condition for one-site translation-invariant spin-$1/2$ chains leaves four genuine three-site models: two FFD-type integrable models and two mixed models that are nonintegrable away from special parameter values, as summarized in Table~\ref{tab:minimal_models_summary}. The corresponding frustration graphs give a complementary explanation of this conserved-charge-based classification.

The remainder of the paper is organized as follows. 
In Sec.~\ref{Sec:integrable_nonintegrable_conditions}, we introduce the notation and summarize the analytical criteria for integrability and nonintegrability used throughout the paper. 
Section~\ref{sec:nonintegrable_test_deform_Fredkin_model} applies these criteria to the deformed Fredkin chain and proves its nonintegrability. 
Building on this result, Sec.~\ref{sec:minimal_models_three_site} constructs the minimal composite models and derives the unblocking condition for one-site translation-invariant three-site spin-$1/2$ chains. Section~\ref{sec:integrability_minimal_composite} classifies the resulting minimal models. Finally, Sec.~\ref{Sec:test_frustration_graph} gives the complementary frustration-graph analysis, and Sec.~\ref{sec:conclusion} contains the conclusions and outlook. The appendices provide supplementary analyses of the theorems presented in the main text.

\section{Integrability and nonintegrability criteria}
\label{Sec:integrable_nonintegrable_conditions}

This section collects the definitions and criteria used throughout the paper.  We first fix the notation and the local-operator spaces in Sec.~\ref{Sec:Hokkyo_model}.  We then review Hokkyo's nonintegrability criterion \cite{hokkyo2025rigorous} in Sec.~\ref{Sec:sufficient_nonintegrability} and the Reshetikhin integrability criterion \cite{hokkyo2025integrability} in Sec.~\ref{Sec:Reshetikhin_condition}.  Finally, Sec.~\ref{Sec:frustration_graph} summarizes the frustration-graph method used as a complementary integrability diagnostic \cite{elman2021free_behind}.

%We work with a charge-based notion of integrability: an interacting spin chain is called integrable if it possesses an infinite family of nontrivial local conserved charges.  The main nonintegrability tool is Hokkyo's rigorous criterion \cite{hokkyo2025rigorous}, which reduces the absence of higher local conserved charges to a finite calculation involving admissible three-local operators, provided that two structural assumptions, injectivity and the two-local conservation condition, are satisfied.  In the opposite direction, we use the recent proof that a Reshetikhin, or boost-generated, three-local conservation law implies an infinite tower of commuting local charges \cite{hokkyo2025integrability}.  Finally, the frustration-graph method \cite{elman2021free_behind} will be used as a complementary diagnostic: it gives a sufficient graphical condition for free-fermion solvability and integrability, and will help clarify why the minimal models found below separate into integrable and nonintegrable classes.

\subsection{Notation and general setting}
\label{Sec:Hokkyo_model}

There is no universally accepted definition of quantum integrability \cite{caux2011remarks}.  In this work we use the following operational definition: a translationally invariant local Hamiltonian is integrable if it admits infinitely many linearly independent nontrivial local conserved charges.  The Hamiltonian itself is conserved, but it is regarded as a trivial local conserved charge.  Operators whose range is larger than half of the finite chain are not counted as local, since such operators may have no well-defined thermodynamic-limit interpretation.

Consider a chain of $N$ sites with periodic boundary conditions.  Let $\mathcal A_j$ be the one-site operator space on site $j$, and let $\mathcal A_{0,j}\subset\mathcal A_j$ be its traceless subspace.  The strictly local density space of length $\ell$ starting at $j$ is denoted by $\mathcal A_j^{(\ell)}$.  For $\ell=1$, $\mathcal A_j^{(1)}=\mathcal A_{0,j}$.  For $\ell\ge2$, $\mathcal A_j^{(\ell)}$ is spanned by operator strings
\begin{equation}
    \mathbf A_j^{(\ell)}=A_jA_{j+1}\cdots A_{j+\ell-1}.
\end{equation}
Here the endpoint operators satisfy $A_j\in\mathcal A_{0,j}$ and $A_{j+\ell-1}\in\mathcal A_{0,j+\ell-1}$, while the intermediate operators satisfy $A_{j+r}\in\mathcal A_{j+r}$ for $1\le r\le \ell-2$.
The corresponding space of strictly $\ell$-local translationally invariant operators is
\begin{equation}
    \mathcal B^{(\ell)}
    =\operatorname{span}
    \left\{\sum_{j=1}^{N}\mathbf A_j^{(\ell)}
    \;\middle|\; \mathbf A_j^{(\ell)}\in\mathcal A_j^{(\ell)}\right\},
\end{equation}
with $\mathcal B^{(0)}=\mathbb C I$. Thus $\mathcal A_j^{(\ell)}$ denotes a density space with length $\ell$, whereas $\mathcal B^{(\ell)}$ denotes the translationally invariant operator space obtained after summing over $j$. 

The length of a translationally invariant local operator $O$ is
\begin{equation}
    \operatorname{len}(O)
    =\min\left\{k\middle|\;
    O\in\bigoplus_{\ell=0}^{k}\mathcal B^{(\ell)}\right\}.
\end{equation}
A translationally invariant $k$-local operator has the form
\begin{equation}
    Q^{(k)}=\sum_{j=1}^{N}\sum_{\ell=1}^{k}q_{\mathbf A}\mathbf A_j^{(\ell)}, 
\end{equation}
where $q_{\mathbf A}$ denotes the coefficient of the corresponding local basis density.  If only terms with $\ell=k$ are present, the operator will be called strictly $k$-local.

The nearest-neighbor Hamiltonians to which Hokkyo's criterion directly applies are written as
\begin{equation}
\label{H_general_form}
    H=H^{(2)}+H^{(1)},\qquad
    H^{(\ell)}=\sum_{j=1}^{N}h_j^{(\ell)} ,
\end{equation}
where $h_j^{(2)}$ is a two-site interaction density and $h_j^{(1)}$ is an on-site term.  Since such a Hamiltonian is at most two-local, conserved charges of length larger than two are the nontrivial charges relevant to integrability.

We shall frequently use a compact column notation for Pauli-string commutators.  For example,
\begin{equation}
    [X_jY_{j+1}Z_{j+2},Y_{j+2}Z_{j+3}]
    =-2iX_jY_{j+1}X_{j+2}Z_{j+3}
\end{equation}
will be represented as
\begin{equation}
\begin{matrix}
& X_j & Y_{j+1} & Z_{j+2} & \\
&     &         & Y_{j+2} & Z_{j+3} \\ \hline
- & X_j & Y_{j+1} & X_{j+2} & Z_{j+3}.
\end{matrix}
\end{equation}
The common factor $2i$ is suppressed, and only the sign and the Pauli string are displayed.  Unless otherwise stated, the first row denotes the operator being tested and the second row denotes a Hamiltonian density term.

\subsection{Hokkyo's criterion: a sufficient condition for nonintegrability}
\label{Sec:sufficient_nonintegrability}

We now recall the nonintegrability criterion of Hokkyo \cite{hokkyo2025rigorous}.  The criterion applies to Hamiltonians of the form \eqref{H_general_form} and rests on two key assumptions.

The first assumption is injectivity.  It requires that no nonzero traceless one-site operator, acting on either side of the two-site density, commute with the interaction.
\begin{subequations}
\label{eq:injectivity}
\begin{align}
    \forall\,0\neq A_j\in\mathcal A_{0,j},\qquad
    &[A_j,h_j^{(2)}]\neq0,\\
    \forall\,0\neq A_{j+1}\in\mathcal A_{0,j+1},\qquad
    &[A_{j+1},h_j^{(2)}]\neq0.
\end{align}
\end{subequations}
This condition excludes singular interactions such as a pure Ising coupling.  For example, in a spin-$1/2$ nearest-neighbor chain, $XX+YY$ satisfies injectivity whereas $ZZ$ alone does not.

The second assumption is the two-local conservation condition.  It restricts the strictly two-local operators whose commutator with the Hamiltonian does not generate length-three terms.  We first define
\begin{equation}
\label{eq:B_le^k}
    \mathcal B_{\le}^{(k)}=\left\{O\in\mathcal B^{(k)}\;\middle|\;
    \operatorname{len}([O,H])\le k\right\}.
\end{equation}
Thus $O\in\mathcal B_{\le}^{(k)}$ means that all potentially generated $(k+1)$-local terms from the commutator of a $k$-local operator with the Hamiltonian cancel.  Since $H$ is at most two-local, the commutator $[O,H]$ can generate terms of length at most $k+1$.  The two-local conservation condition is
\begin{equation}
\label{eq:2-local_condition}
    \mathcal B_{\le}^{(2)}=\mathbb C H^{(2)}.
\end{equation}
When the above condition holds, the only possible three-local conserved charge has the three-local part proportional to $[h_j^{(2)},h_{j+1}^{(2)}]$.  Under these two assumptions, Hokkyo's criterion reads as follows \cite{hokkyo2025rigorous}.

\begin{theorem}
\label{theorem:Hokkyo_main_result}
For a nearest-neighbor Hamiltonian of the form \eqref{H_general_form} satisfying injectivity \eqref{eq:injectivity} and the two-local conservation condition \eqref{eq:2-local_condition}, suppose that there is no three-local operator $Q^{(3)}$ such that
\begin{equation}
\label{eq:Hokkyo_nonintegrable_inequality}
    \operatorname{len}([Q^{(3)},H])\le2.
\end{equation}
Then the model has no $k$-local conserved charge for any $3\le k\le N/2$.
\end{theorem}

The theorem provides a useful simplification for proving the absence of local conserved charges.  Instead of excluding conserved quantities separately at every range, as in the general method of Shiraishi \cite{shiraishi2019proof}, one only has to solve a finite three-local obstruction problem after the two structural assumptions have been verified. For more technical details about Hokkyo's criterion, see Appendix~\ref{App:Hokkyo_sets}.

%\paragraph{A practical pre-screening for the two-local conservation condition.}
In practice it is useful to first identify which two-local Pauli strings could possibly belong to $\mathcal B_{\le}^{(2)}$.  For the injective Pauli-string Hamiltonians considered below, let $\mathcal B_l$ and $\mathcal B_r$ be the sets of Pauli strings appearing, respectively, on the left and right sites of $h_j^{(2)}$.  If a two-site string has a left factor outside $\mathcal B_l$ or a right factor outside $\mathcal B_r$, it cannot cancel the leading length-three terms generated by commuting with $H$.  Hence
\begin{equation}
    \mathcal B_{\le}^{(2)}\subset \mathcal B_{lr}^{(2)},
\end{equation}
where $\mathcal B_{lr}^{(2)}$ is the set of strings with left factor in $\mathcal B_l$ and right factor in $\mathcal B_r$.

For candidate two-local conserved charges in $\mathcal B_{\le}^{(2)}$, the commutators that generate length-three terms must cancel, giving the condition
\begin{equation}\label{exclude_2_local}
    \sum_{O\in\mathcal B_{\le}^{(2)}}
    \left([O_j,h_{j+1}^{(2)}]+[O_{j+1},h_j^{(2)}]\right)=0.
\end{equation}
Here $O_j$ denotes the local density operator acting on sites $(j,j+1)$.  We shall use this elementary but efficient test below.
% If $A_jA_{j+1}$ and $A'_{j+1}A'_{j+2}$ are both to contribute to $\mathcal B_{\le}^{(2)}$, their leading commutators must be able to cancel, which requires relations of the form
% \begin{equation}
% \label{exclude_2_local}
%     [A_jA_{j+1},h_{j+1}h_{j+2}]
%     +[A'_{j+1}A'_{j+2},h'_jh'_{j+1}]=0 .
% \end{equation}

\subsection{Reshetikhin condition and a criterion for integrability}
\label{Sec:Reshetikhin_condition}

For nearest-neighbor Hamiltonians of the form \eqref{H_general_form}, Yang-Baxter integrability implies the existence of a three-local conserved charge of Reshetikhin form \cite{KulishSklyanin1982QSTM}.  In the notation used above, the relevant conserved charge ansatz is
\begin{equation}
\label{Reshetikhin}
    Q^{(3)}=
    \sum_j\left([h_j^{(2)},h_{j+1}^{(2)}]+\mathbf A_j^{(2)}\right),
\end{equation}
where $\mathbf A_j^{(2)}$ is a two-local correction term.  This condition has long served as a practical integrability test, and is closely related to the conjecture of Grabowski and Mathieu that a nontrivial three-local conserved charge signals integrability in translation-invariant nearest-neighbor chains \cite{grabowski1995integrability}.

The recent result of Hokkyo \cite{hokkyo2025integrability} gives a rigorous implication in precisely the direction needed here: if the boost-generated three-local charge is conserved, then it generates an infinite family of mutually commuting local charges.  Formally we have the following theorem. 

\begin{theorem}
\label{theorem:Reshetikhin_rigorous}
For a nearest-neighbor Hamiltonian of the form \eqref{H_general_form}, suppose that the Reshetikhin charge \eqref{Reshetikhin} is conserved.  Then the corresponding construction produces an infinite hierarchy of commuting local conserved charges.
\end{theorem}

Together with Theorem~\ref{theorem:Hokkyo_main_result}, this result gives a practical dichotomy for the models studied below.  If the Reshetikhin equation can be solved, integrability follows.  If injectivity and the two-local conservation condition hold but no admissible three-local charge exists, nonintegrability follows.  Related dichotomy results have recently been obtained for isotropic spin chains \cite{Shiraishi2025DichotomyTS}.  This dichotomy should not, however, be understood as universal; recent work on non-Hermitian bosonic chains shows that such all-or-nothing alternatives can fail \cite{Yamaguchi2026ViolatingAllOrNothing}.

\subsection{Frustration-graph method as an integrability diagnostic}
\label{Sec:frustration_graph}

The frustration-graph method was originally developed to characterize free-fermion solvability in Pauli-string Hamiltonians \cite{chapman2020characterization}.  In the free-fermions-in-disguise framework, it was further shown to provide a sufficient graphical condition for integrability \cite{elman2021free_behind}.  We use it here only as a supplementary diagnostic, especially for identifying free-fermion-type integrable structures, not as the main proof of nonintegrability.

The construction is simple.  Each Pauli-string term in the Hamiltonian is represented by a vertex, and two vertices are connected by an edge if the corresponding Pauli strings anticommute.  An induced subgraph is obtained by selecting a subset of vertices together with all edges among them.  The two forbidden induced subgraphs relevant here are claws and even holes, illustrated in Fig.~\ref{fig:forbidden_subgraphs}.  A claw is a $K_{1,3}$ graph, namely a central vertex connected to three mutually disconnected leaves.  An even hole is an induced cycle of even length with no chord.

\begin{figure}[t]
\centering
\begin{tikzpicture}
    \coordinate (A) at (0,0);
    \coordinate (B) at (2,0);
    \coordinate (C) at (1,1.732);
    \coordinate (O) at (1,0.577);
    \draw[thick] (O) -- (A);
    \draw[thick] (O) -- (B);
    \draw[thick] (O) -- (C);
    \draw[fill] (A) circle [radius=0.05];
    \draw[fill] (B) circle [radius=0.05];
    \draw[fill] (C) circle [radius=0.05];
    \draw[fill] (O) circle [radius=0.05];
\end{tikzpicture}
\qquad\qquad
\begin{tikzpicture}
    \coordinate (A) at (0,0);
    \coordinate (B) at (1.5,0);
    \coordinate (C) at (1.5,1.5);
    \coordinate (D) at (0,1.5);
    \draw[thick] (A) -- (B) -- (C) -- (D) -- cycle;
    \draw[fill] (A) circle [radius=0.05];
    \draw[fill] (B) circle [radius=0.05];
    \draw[fill] (C) circle [radius=0.05];
    \draw[fill] (D) circle [radius=0.05];
\end{tikzpicture}
\qquad\qquad
\begin{tikzpicture}
    \foreach \i in {0,...,5} {
        \coordinate (p\i) at (60*\i:1);
    }
    \draw[thick] (p0) -- (p1) -- (p2) -- (p3) -- (p4) -- (p5) -- cycle;
    \foreach \i in {0,...,5} {
        \draw[fill] (p\i) circle [radius=0.05];
    }
\end{tikzpicture}
\caption{Forbidden induced subgraphs in the frustration-graph approach: a claw (left) and even holes with four and six vertices (middle and right).}
\label{fig:forbidden_subgraphs}
\end{figure}

\begin{theorem}
\label{theorem:frustration_graph}
If the frustration graph of a Pauli-string Hamiltonian is claw-free, then the Hamiltonian is integrable.  If its frustration graph is both claw-free and even-hole-free, then the Hamiltonian admits a free-fermion solution.
\end{theorem}

Here claw-free means that no induced subgraph is a claw; similarly, even-hole-free means that no induced subgraph is an even hole.  The criterion is sufficient but not necessary \cite{chapman2023unified}.  In particular, some integrable models have frustration graphs containing claws or even holes \cite{fendley2024fbeyond,fukai2025free_claw}.  In the present work the graph analysis should therefore be read as a structural explanation of the charge-based classification, rather than as an independent proof of nonintegrability.

\section{Integrability test on the deformed Fredkin spin chain}
\label{sec:nonintegrable_test_deform_Fredkin_model}

This section applies the above criteria to the deformed Fredkin chain.  We first introduce the periodic spin-$1/2$ deformed Fredkin Hamiltonian and rewrite it as a nearest-neighbor model on composite sites in Sec.~\ref{deformed_Fredkin_model}.  We then prove in Sec.~\ref{integrability_test_deform} that this composite Hamiltonian satisfies the assumptions of Hokkyo's criterion for any nonzero deformation parameter.  It follows that the original deformed Fredkin chain is nonintegrable.  The proof also isolates a small set of three-site Pauli structures that will later serve as building blocks for minimal nonintegrable models.

\subsection{Hamiltonian of the deformed Fredkin spin chain}
\label{deformed_Fredkin_model}

The deformed Fredkin spin chain has the Hamiltonian \cite{salberger2017deformed}
\begin{equation}
    H(s,t)=H_F(s,t)+H_X(s)+H_\partial(s),
\end{equation}
where $s$ is the color index and $t>0$ is the deformation parameter.  The term $H_F(s,t)$ is the bulk deformed Fredkin Hamiltonian, $H_X(s)$ appears in the colored generalizations with $s>1$, and $H_\partial(s)$ is the boundary contribution.  We focus on the uncolored periodic chain, corresponding to $s=1$.  In this case the color-dependent terms are absent, and periodic boundary conditions give
\begin{equation}
\label{HF_s=1/2}
    H(t)=\sum_{j=1}^{N}\left(
    |\phi_{j,+}\rangle\langle\phi_{j,+}|
    +|\phi_{j,-}\rangle\langle\phi_{j,-}|
    \right),
\end{equation}
with
\begin{subequations}
\begin{align}
    |\phi_{j,+}\rangle
    &=\frac{1}{\sqrt{1+t^2}}
    |\uparrow_{j-1}\rangle\otimes
    \left(|\uparrow_j\downarrow_{j+1}\rangle
    -t|\downarrow_j\uparrow_{j+1}\rangle\right),\\
    |\phi_{j,-}\rangle
    &=\frac{1}{\sqrt{1+t^2}}
    \left(|\uparrow_{j-1}\downarrow_j\rangle
    -t|\downarrow_{j-1}\uparrow_j\rangle\right)
    \otimes|\downarrow_{j+1}\rangle.
\end{align}
\end{subequations}
Up to an overall factor $1/(4(1+t^2))$, additive identity terms, and a one-site term that vanishes after the periodic translational sum, the Pauli form of the local density is
\begin{equation}
\label{Fredkin_density}
    h_j(t)\propto h_j^{(2)}(t)+h_j^{(3)}(t),
\end{equation}
where
\begin{subequations}
\begin{multline}
\label{Fredkin_density_2}
    h_j^{(2)}(t)=
    -tX_jX_{j+1}I_{j+2}-tI_jX_{j+1}X_{j+2}\\
    -tY_jY_{j+1}I_{j+2}-tI_jY_{j+1}Y_{j+2}\\
    -t^2Z_jZ_{j+1}I_{j+2}
    -t^2I_jZ_{j+1}Z_{j+2}.
\end{multline}
\begin{multline}
\label{Fredkin_density_3}
    h_j^{(3)}(t)=
    -tZ_jX_{j+1}X_{j+2}-tZ_jY_{j+1}Y_{j+2}\\
    +tX_jX_{j+1}Z_{j+2}
    +tY_jY_{j+1}Z_{j+2}\\
    +(t^2-1)Z_jI_{j+1}Z_{j+2}.
\end{multline}
\end{subequations}
At $t=1$ this reduces to the original Fredkin spin chain, whose nonintegrability was proved in Ref.~\cite{fan2025absence}.  The deformed case is structurally richer because of the additional Ising-type term proportional to $t^2-1$.

The important observation for the present work is that the terms in Eq.~\eqref{Fredkin_density} can be viewed as coupled simple integrable constituents.  The two-site part contains an XXZ-type interaction, while the three-site terms $XXZ$, $YYZ$, $ZXX$, and $ZYY$ are FFD-type Pauli strings \cite{fendley2019free_in}.  The nonintegrability shown below is therefore not caused by a complicated individual term, but by the way these integrable pieces are coupled.

The nonintegrability criterion in Theorem~\ref{theorem:Hokkyo_main_result} is formulated for nearest-neighbor Hamiltonians.  We therefore group two neighboring physical spins into one composite site, $n=(2j-1,2j)$, so that a three-site spin-$1/2$ Hamiltonian becomes a nearest-neighbor Hamiltonian with local Hilbert space of dimension four.  This blocking is a local algebra isomorphism 
\begin{equation}
    \bigotimes_{j=1}^{N}\mathbb C^2
    \simeq
    \bigotimes_{n=1}^{L}(\mathbb C^2\otimes\mathbb C^2),
\end{equation}
therefore preserves commutators and conservation laws.  In the composite-site representation, we use the two-qubit basis $\{I,X,Y,Z\}^{\otimes 2}$. The same blocking idea is standard in treating finite-range interactions as nearest-neighbor interactions with enlarged local space, and is closely related to the medium-range viewpoint used in integrability studies \cite{gombor2021integrable}.

For even $N=2L$, the deformed Fredkin Hamiltonian in the composite-site representation can be written as
\begin{equation}
\label{deformed_equivalent}
    H(t)=\sum_{n=1}^{L}\left(h_n^{(2)}(t)+h_n^{(1)}(t)\right),
\end{equation}
where
\begin{subequations}
\begin{equation}
\label{deformed_equivalent_1}
    h_n^{(1)}(t)=
    -2t(XX)_n-2t(YY)_n-2t^2(ZZ)_n,
\end{equation}
\begin{multline}
\label{deformed_equivalent_2}
    h_n^{(2)}(t)
    =(t^2-1)(IZ)_n(IZ)_{n+1}
    +(t^2-1)(ZI)_n(ZI)_{n+1}
    \\
    -2t(IX)_n(XI)_{n+1}
    -2t(IY)_n(YI)_{n+1}
    \\
    -2t^2(IZ)_n(ZI)_{n+1}
    +t(IX)_n(XZ)_{n+1}
    \\
    +t(IY)_n(YZ)_{n+1}
    +t(XX)_n(ZI)_{n+1}
    \\
    +t(YY)_n(ZI)_{n+1}
    -t(IZ)_n(XX)_{n+1}
    \\
    -t(IZ)_n(YY)_{n+1}
    -t(ZX)_n(XI)_{n+1}
    \\
    -t(ZY)_n(YI)_{n+1}.
\end{multline}
\end{subequations}
Here $h^{(2)}_n$ denotes the genuine nearest-neighbor part in the composite chain, while $h^{(1)}_n$ denotes the on-site term in the form required by Eq.~\eqref{H_general_form}. 

\subsection{Nonintegrability of the deformed Fredkin spin chain}
\label{integrability_test_deform}

We first verify the assumptions required for Theorem~\ref{theorem:Hokkyo_main_result}.  A direct finite-dimensional check of the left and right Pauli operators appearing in Eq.~\eqref{deformed_equivalent_2} shows that the interaction is injective for $t\neq0$.  More explicitly, write a single-composite-site operator as a linear combination of the fifteen nonidentity two-qubit Pauli strings.  Imposing that this operator commute with $h_n^{(2)}(t)$ from the left or from the right gives, for $t\neq0$, homogeneous linear systems whose only solution is the zero operator.  Thus no nonzero traceless single-composite-site operator commutes with the interaction $h_n^{(2)}(t)$ on either side.

It remains to prove the two-local conservation condition in Eq.~\eqref{eq:2-local_condition}.  We state this as a lemma.

\begin{lemma}
\label{theorem:fredkin_B2_condition}
For the deformed Fredkin spin chain in the composite-site representation \eqref{deformed_equivalent}, with $t\neq0$, the two-local conservation condition $\mathcal B_{\le}^{(2)}=\mathbb C H^{(2)}$ holds.
\end{lemma}

\begin{proof}
For the two-site density \eqref{deformed_equivalent_2}, the left and right Pauli operators are
\begin{subequations}
\begin{align}
   \mathcal B_l&=\{IX,IY,IZ,XX,YY,ZI,ZX,ZY\},\\
   \mathcal B_r&=\{IZ,XI,XX,XZ,YI,YY,YZ,ZI\}.
\end{align}
\end{subequations}
Hence any possible element of $\mathcal B_{\le}^{(2)}$ must lie in the $64$-dimensional span $\mathcal B_{lr}^{(2)}$.  We can first eliminate all non-Hamiltonian operators by using the commutator cancellation relation \eqref{exclude_2_local}.  As an illustration, consider $(YY)(XX)$ and commute it with the Hamiltonian term $(IY)(YI)$:
\begin{equation}
\begin{matrix}
  & (Y & Y) & (X & X) &  & \\
  &  &  & (I & Y) & (Y & I) \\ \hline
  & (Y & Y) & (X & Z) & (Y & I).
\end{matrix}
\end{equation}
The resulting length-three string $(YY)(XZ)(YI)$ is not produced by the commutator of any other operator of $\mathcal B_{lr}^{(2)}$ with the Hamiltonian.  Therefore $(YY)(XX)$ cannot belong to $\mathcal B_{\le}^{(2)}$.  Applying the same test to the remaining non-Hamiltonian operators leaves only the thirteen two-site terms already present in $H^{(2)}$.

We then determine their allowed coefficients.  Let the coefficients of $(IX)(XZ)$ and $(XX)(ZI)$ be $q_{IXXZ}$ and $q_{XXZI}$.  Their commutators with the Hamiltonian give
\begin{equation}
\begin{gathered}
\begin{matrix}
q_{IXXZ}&(I&X)&(X&Z)&&\\
t&&&(X&X)&(Z&I)\\ \hline
q_{IXXZ}t&(I&X)&(I&Y)&(Z&I),
\end{matrix}
\\[1ex]
\begin{matrix}
q_{XXZI}&&&(X&X)&(Z&I)\\
t&(I&X)&(X&Z)&&\\ \hline
-q_{XXZI}t&(I&X)&(I&Y)&(Z&I).
\end{matrix}
\end{gathered}
\end{equation}
Cancellation requires $q_{IXXZ}=q_{XXZI}$.  Similarly, comparing $(IZ)(ZI)$ and $(XX)(ZI)$ gives
\begin{equation}
\begin{gathered}
\begin{matrix}
q_{IZZI}&(I&Z)&(Z&I)&&\\
t&&&(X&X)&(Z&I)\\ \hline
q_{IZZI}t&(I&Z)&(Y&X)&(Z&I),
\end{matrix}
\\[1ex]
\begin{matrix}
q_{XXZI}&&&(X&X)&(Z&I)\\
-2t^2&(I&Z)&(Z&I)&&\\ \hline
2q_{XXZI}t^2&(I&Z)&(Y&X)&(Z&I),
\end{matrix}
\end{gathered}
\end{equation}
which implies $q_{IZZI}=-2tq_{XXZI}$.  Continuing this finite elimination fixes all coefficients relative to the coefficient $c$ of $(IX)(XZ)$ and gives
\begin{multline}
\label{B=CH}
    g^{(2)}=
    \frac{t^2-1}{t}c(IZ)(IZ)
    +\frac{t^2-1}{t}c(ZI)(ZI)
    \\
    -2c(IX)(XI)-2c(IY)(YI)
    \\
    -2tc(IZ)(ZI)+c(IX)(XZ)
    \\
    +c(IY)(YZ)+c(XX)(ZI)
    \\
    +c(YY)(ZI)-c(IZ)(XX)
    \\
    -c(IZ)(YY)-c(ZX)(XI)
    \\
    -c(ZY)(YI).
\end{multline}
Setting $c=t$ gives $g^{(2)}=h^{(2)}$.  After summing over translations, this is precisely the condition $\mathcal B_{\le}^{(2)}=\mathbb C H^{(2)}$.
\end{proof}

We can now apply Hokkyo's criterion.

\begin{theorem}
\label{theorem:nonintegrable_deformed_Fredkin}
For $t\neq0$, the deformed Fredkin Hamiltonian has no conserved charges with composite range $3\le r\le L/2$.  Consequently, it is nonintegrable in the sense of admitting no infinite tower of nontrivial local conserved charges.
\end{theorem}

\begin{proof}
By Lemma~\ref{theorem:fredkin_B2_condition}, the only possible three-local part of an admissible conserved charge is $[h_n^{(2)},h_{n+1}^{(2)}]$.  We denote this local commutator by
\begin{equation}
\label{eq:iota_2}
    \iota_2(h^{(2)})_n=[h_n^{(2)},h_{n+1}^{(2)}].
\end{equation}
It is therefore enough to show that no choice of the two-local density $\mathbf A_n^{(2)}$ can make
\begin{equation}
    \operatorname{len}\left(\left[\sum_n\left(\iota_2(h^{(2)})_n+\mathbf A^{(2)}_n\right),H\right]\right)\le2.
\end{equation}
We exhibit a three-local Pauli string whose coefficient cannot be cancelled.

The following four contributions appear in $\iota_2(h^{(2)})_n$:
\begin{subequations}
\begin{equation}
\begin{matrix}
t&(I&X)&(X&Z)&&\\
t&&&(I&Y)&(Y&Z)\\ \hline
-t^2&(I&X)&(X&X)&(Y&Z),
\end{matrix}
\end{equation}
\begin{equation}
\begin{matrix}
t&(I&X)&(X&Z)&&\\
t&&&(I&X)&(X&Z)\\ \hline
t^2&(I&X)&(X&Y)&(X&Z),
\end{matrix}
\end{equation}
\begin{equation}
\begin{matrix}
t&(I&X)&(X&Z)&&\\
-2t&&&(I&X)&(X&I)\\ \hline
-2t^2&(I&X)&(X&Y)&(X&I),
\end{matrix}
\end{equation}
\begin{equation}
\begin{matrix}
t&(I&X)&(X&Z)&&\\
-2t&&&(I&Y)&(Y&I)\\ \hline
2t^2&(I&X)&(X&X)&(Y&I).
\end{matrix}
\end{equation}
\end{subequations}
Commuting these four three-local strings with the corresponding two-site Hamiltonian terms gives the same output string
\begin{subequations}
\begin{equation}
\begin{matrix}
-t^2&(I&X)&(X&X)&(Y&Z)\\
-2t&&&(I&X)&(X&I)\\ \hline
-2t^3&(I&X)&(X&I)&(Z&Z),
\end{matrix}
\end{equation}
\begin{equation}
\begin{matrix}
t^2&(I&X)&(X&Y)&(X&Z)\\
-2t&&&(I&Y)&(Y&I)\\ \hline
-2t^3&(I&X)&(X&I)&(Z&Z),
\end{matrix}
\end{equation}
\begin{equation}
\begin{matrix}
-2t^2&(I&X)&(X&Y)&(X&I)\\
t&&&(I&Y)&(Y&Z)\\ \hline
-2t^3&(I&X)&(X&I)&(Z&Z),
\end{matrix}
\end{equation}
\begin{equation}
\begin{matrix}
2t^2&(I&X)&(X&X)&(Y&I)\\
t&&&(I&X)&(X&Z)\\ \hline
-2t^3&(I&X)&(X&I)&(Z&Z).
\end{matrix}
\end{equation}
\end{subequations}
Thus $[\iota_2(h^{(2)})_n,H]$ contains $-8t^3(IX)(XI)(ZZ)$.  This term cannot be generated by $[\mathbf A^{(2)}_n,h^{(2)}_{n'}]$ with one-site overlap: no right factor of the two-site density can supply the necessary $ZZ$ on the last composite site, and the alternative overlap would require a one-site commutator producing the middle factor $XI$, which is absent.  It is also not generated by the one-site part $h^{(1)}$.  Therefore no choice of $\mathbf A^{(2)}_n$ can cancel this length-three term.  Theorem~\ref{theorem:Hokkyo_main_result} then rules out all composite-local conserved charges of range $3\le r\le L/2$.
\end{proof}

The obstruction above is generated by a small subset of terms from the Fredkin Hamiltonian density, but the cancellation statement is made for the full deformed Fredkin Hamiltonian: the remaining terms neither produce the same Pauli string with the opposite coefficient nor allow any two-local correction to generate a cancelling contribution.  This is where the coupled integrable building blocks become genuinely nonintegrable.  The case $t=1$ is included and recovers the original Fredkin chain, while $t=0$ is excluded because the Hamiltonian degenerates to an Ising-type model and no longer satisfies the injective structure.  Compared with the order-by-order proof used in Ref.~\cite{fan2025absence}, the present proof is much simpler because Hokkyo's criterion reduces the problem to a single three-local obstruction.

\section{Minimal models with three-site interactions}
\label{sec:minimal_models_three_site}

Having proved nonintegrability of the deformed Fredkin chain, we now ask how small a Fredkin-derived three-site spin chain can be while still retaining the injective structure needed for the nonintegrability test.  We first formulate the minimal injectivity condition in Sec.~\ref{subsec:minimal_models_injectivity}.  We then apply this condition to the deformed Fredkin density in Sec.~\ref{subsec:minimal_models_from_deformed_fredkin}, obtaining sixteen minimal injective models. Finally, Sec.~\ref{subsec:minimal_models_qubit_representation} imposes the unblocking constraint and selects the models that give genuine one-site translation-invariant three-site spin-$1/2$ Hamiltonians.

\subsection{Minimal models satisfying injectivity}
\label{subsec:minimal_models_injectivity}

For a nearest-neighbor density $h_j^{(2)}$, injectivity means that no nonzero traceless one-site operator commutes with the interaction from either side, as in Eq.~\eqref{eq:injectivity}.  We first recall the simplest spin-$1/2$ case.

\begin{theorem}[Minimal injective density for spin-$1/2$]
\label{theorem:spin_half_minimal_injectivity}
For a spin-$1/2$ nearest-neighbor Hamiltonian density
\begin{equation}
\label{eq:spin_half_two_term_density}
    h_j^{(2)}=c_1h_{a,j}h_{b,j+1}+c_2h_{c,j}h_{d,j+1},
\end{equation}
with $c_1c_2\neq0$, where $h_a,h_b,h_c,h_d$ are nonzero traceless Hermitian one-site operators, injectivity holds if and only if
\begin{equation}
\label{eq:spin_half_injectivity_condition}
    [h_a,h_c]\neq0,
    \qquad
    [h_b,h_d]\neq0.
\end{equation}
\end{theorem}

\begin{proof}
For a traceless operator $A$ on the left site,
\begin{equation}
    [A_j,h_j^{(2)}]
    =c_1[A,h_a]_jh_{b,j+1}+c_2[A,h_c]_jh_{d,j+1}.
\end{equation}
If $[h_b,h_d]\neq0$, then $h_b$ and $h_d$ are linearly independent, so the vanishing of this commutator implies $[A,h_a]=[A,h_c]=0$.  Writing $A = \vec x\cdot \vec\sigma$, $h_a=\vec a\cdot \vec\sigma$, and $h_c=\vec c\cdot \vec\sigma$, where $\vec\sigma=(X,Y,Z)$ and $\vec x,\vec a,\vec c\in\mathbb R^3$, we have
\begin{equation}
    [A,h_a]=2i(\vec x\times \vec a)\cdot \vec\sigma,
\end{equation}
with an analogous expression for $h_c$.  If $[h_a,h_c]\neq0$, then $\vec a$ and $\vec c$ are not parallel, and the only vector $\vec x$ parallel to both is $\vec x=0$.  Hence left injectivity holds.  The right condition is identical.

Conversely, if $[h_a,h_c]=0$, then the two traceless spin-$1/2$ operators are proportional, and choosing $A=h_a$ gives a nonzero traceless operator commuting with both left operators. Thus left injectivity fails.  The case $[h_b,h_d]=0$ similarly gives failure on the right.  Finally, one product term can never be injective, since its own left and right operators commute with it.
\end{proof}

For the Fredkin spin chain, each composite site contains two qubits.  It is therefore useful to formulate the injectivity condition for a general $m$-qubit composite site.  Let
\begin{equation}
    \mathcal P_m=\{I,X,Y,Z\}^{\otimes m}
\end{equation}
be the Pauli-string basis.  We use the standard binary representation of Pauli strings.  On each qubit, the two binary variables $(x_\ell,z_\ell)$ record whether an $X_\ell$ and a $Z_\ell$ operator are present.  Thus $(0,0)$ represents $I_\ell$, $(1,0)$ represents $X_\ell$, $(0,1)$ represents $Z_\ell$, and $(1,1)$ represents $Y_\ell$ up to a phase convention.  Since injectivity depends only on whether commutators vanish, these phases play no role.  Accordingly, for a Pauli string $P\in\mathcal P_m$ we write
\begin{equation}
    P=X_1^{x_1}\cdots X_m^{x_m}Z_1^{z_1}\cdots Z_m^{z_m},
\end{equation}
where equality is understood up to an irrelevant overall phase.  The Pauli string is then represented by the binary vector
\begin{equation}
    v(P)=(x_1,\ldots,x_m,z_1,\ldots,z_m)\in\mathbb F_2^{2m}.
\end{equation}
This is the binary symplectic notation used in the stabilizer formalism \cite{aaronson2004improved}. Here it is used only to track commutation relations.  With this convention, the commutation relation of two Pauli strings can be expressed as
\begin{equation}
\label{eq:symplectic_product_m_qubit}
    \langle v(P),v(P')\rangle_{\rm sp}
    =\sum_{\ell=1}^{m}(x_\ell z'_\ell+z_\ell x'_\ell)
    \pmod 2.
\end{equation}
Here the subscript ${\rm sp}$ denotes the binary symplectic product.  Thus $P$ and $P'$ commute when $\langle v(P),v(P')\rangle_{\rm sp}=0$ and anticommute when $\langle v(P),v(P')\rangle_{\rm sp}=1$.  

In what follows, linear independence, span, and full rank refer to the binary label space $v(P)\in\mathbb F_2^{2m}$, with coefficients in $\mathbb F_2$.  This binary span should not be confused with the usual complex linear span of Pauli matrices as operators.  Distinct Pauli strings remain linearly independent in the matrix sense, whereas the rank condition below concerns only their binary commutation labels.  This representation is standard in the Pauli-group literature \cite{planat2007pauli}. 

For spin chains with composite sites, the following theorem gives a necessary and sufficient condition for injectivity in terms of these binary symplectic labels.

\begin{theorem}
\label{theorem:m_qubit_minimal_injectivity}
Consider an $m$-qubit composite nearest-neighbor Hamiltonian density 
\begin{equation}
\label{eq:m_qubit_density}
    h_j^{(2)}=\sum_{\alpha=1}^{2m}c_\alpha P_{\alpha,j}Q_{\alpha,j+1},
    \qquad c_\alpha\neq0,
\end{equation}
where $P_\alpha,Q_\alpha\in\mathcal P_m\setminus\{I^{\otimes m}\}$.  The density is injective if and only if
\begin{equation}
\label{eq:m_qubit_rank_condition}
    \operatorname{rank}_{\mathbb F_2}\operatorname{span}\{v(P_\alpha)\}_{\alpha=1}^{2m}
    =
    \operatorname{rank}_{\mathbb F_2}\operatorname{span}\{v(Q_\alpha)\}_{\alpha=1}^{2m}
    =2m.
\end{equation}
\end{theorem}

\begin{proof}
We prove the left condition; the right one is identical.  Let $A$ be a traceless operator on the left composite site.  Then
\begin{equation}
    [A_j,h_j^{(2)}]
    =\sum_{\alpha=1}^{2m}c_\alpha[A,P_\alpha]_jQ_{\alpha,j+1}.
\end{equation}
If the right Pauli vectors have full rank, then the $Q_\alpha$ are linearly independent over $\mathbb C$.  Therefore the vanishing of the commutator implies $[A,P_\alpha]=0$ for every $\alpha$.

Expand $A$ in the Pauli basis, $A=\sum_{R\in\mathcal P_m}a_RR$.  A Pauli component $R$ commutes with all $P_\alpha$ precisely when
\begin{equation}
    \langle v(R),v(P_\alpha)\rangle_{\rm sp}=0
\end{equation}
for all $\alpha$. Thus the possible nonidentity components of $A$ are labeled by the symplectic orthogonal complement of $\operatorname{span}_{\mathbb F_2}\{v(P_\alpha)\}$.  Because the symplectic form is nondegenerate, this complement is trivial if and only if the binary labels $v(P_\alpha)$ span $\mathbb F_2^{2m}$.  Hence the only common commutant is the identity, and a traceless $A$ must vanish.

Conversely, if the left rank is smaller than $2m$, then there exists a nonzero binary vector in the symplectic orthogonal complement.  The corresponding nonidentity Pauli string commutes with all $P_\alpha$, violating injectivity.  The same argument applies to the right operators.  Since an $r$-term Pauli-string product density has at most $r$ left binary labels and at most $r$ right binary labels, fewer than $2m$ terms cannot have full binary rank on both sides.
\end{proof}

This theorem is a statement about Hamiltonian densities written as sums of Pauli-string product operators.  It does not assert a minimal term count for arbitrary non-Pauli operator decompositions.  In the applications below, $m=2$.  Thus a minimal injective composite density within this Pauli-string class must contain four product terms whose left and right Pauli labels each form a basis of $\mathbb F_2^4$.

\subsection{Minimal models from the deformed Fredkin spin chain}
\label{subsec:minimal_models_from_deformed_fredkin}

We now apply the injectivity rank criterion to the two-site density in Eq.~\eqref{deformed_equivalent_2}.  Write 
\begin{equation}
    h_n^{(2)}(t)=\sum_{\alpha=1}^{13}\gamma_\alpha(t)T_\alpha,
    \qquad
    T_\alpha=(P_\alpha)(Q_\alpha),
\end{equation}
where the shorthand $T_\alpha=(P_\alpha)(Q_\alpha)$ denotes $P_{\alpha,n}Q_{\alpha,n+1}$, with $P_\alpha$ and $Q_\alpha$ acting on the left and right composite sites, respectively.  The thirteen terms and their binary symplectic vectors are listed in Table~\ref{tab:fredkin_terms_binary}.

\begin{table}[t]
\caption{Pauli-string operators in the Hamiltonian density $h_n^{(2)}(t)$ given by Eq.~\eqref{deformed_equivalent_2}.  The binary vectors are written in the order $v(P)=(x_1,x_2,z_1,z_2)$.}
\label{tab:fredkin_terms_binary}
\centering
%\footnotesize
\begin{tabular}{c|c|c|c|c}
\hline\hline
~$\alpha$~ & ~$\gamma_\alpha(t)$~ & ~$T_\alpha=(P_\alpha)(Q_\alpha)$~ & ~$v(P_\alpha)$~ & ~$v(Q_\alpha)$~\\
\hline
1  & $t^2-1$  & $(IZ)(IZ)$ & $\mathtt{0001}$ & $\mathtt{0001}$\\
2  & $t^2-1$  & $(ZI)(ZI)$ & $\mathtt{0010}$ & $\mathtt{0010}$\\
3  & $-2t$    & $(IX)(XI)$ & $\mathtt{0100}$ & $\mathtt{1000}$\\
4  & $-2t$    & $(IY)(YI)$ & $\mathtt{0101}$ & $\mathtt{1010}$\\
5  & $-2t^2$  & $(IZ)(ZI)$ & $\mathtt{0001}$ & $\mathtt{0010}$\\
6  & $t$      & $(IX)(XZ)$ & $\mathtt{0100}$ & $\mathtt{1001}$\\
7  & $t$      & $(IY)(YZ)$ & $\mathtt{0101}$ & $\mathtt{1011}$\\
8  & $t$      & $(XX)(ZI)$ & $\mathtt{1100}$ & $\mathtt{0010}$\\
9  & $t$      & $(YY)(ZI)$ & $\mathtt{1111}$ & $\mathtt{0010}$\\
10 & $-t$     & $(IZ)(XX)$ & $\mathtt{0001}$ & $\mathtt{1100}$\\
11 & $-t$     & $(IZ)(YY)$ & $\mathtt{0001}$ & $\mathtt{1111}$\\
12 & $-t$     & $(ZX)(XI)$ & $\mathtt{0110}$ & $\mathtt{1000}$\\
13 & $-t$     & $(ZY)(YI)$ & $\mathtt{0111}$ & $\mathtt{1010}$\\
\hline\hline
\end{tabular}
\end{table}

Injectivity captures the idea that the local interaction is sufficiently rich to leave no nontrivial one-site commutant.  To locate the minimal boundary of this property, we apply Theorem~\ref{theorem:m_qubit_minimal_injectivity} to four-term subdensities of the deformed Fredkin Hamiltonian. The following theorem gives the sixteen minimal injective models contained in the deformed Fredkin Hamiltonian density.

\begin{theorem}
\label{theorem:minimal_injective_fredkin_four_terms}
Let $S$ be a four-element subset of labels in $\{1,\ldots,13\}$.  Define the Hamiltonian density
\begin{equation}
\label{eq:h_S^2}
    h_S^{(2)}=\sum_{\alpha\in S}g_\alpha T_\alpha,
    \qquad g_\alpha\neq0,
\end{equation}
where the $T_\alpha$ are the Pauli-string operators in Table~\ref{tab:fredkin_terms_binary}.  Then $h_S^{(2)}$ is injective if and only if 
\begin{equation}
\label{eq:sixteen_term_set_condition}
    S=\{\alpha_1,\alpha_2,\alpha_3,\alpha_4\},
\end{equation}
with $\alpha_1\in\{6,7\}$, $\alpha_2\in\{8,9\}$, $\alpha_3\in\{10,11\}$, and $\alpha_4\in\{12,13\}$. 
\end{theorem}

\begin{proof}
For a given label set $S$, let $V_L(S)$ and $V_R(S)$ be the $4\times4$ binary matrices whose rows are $v(P_\alpha)$ and $v(Q_\alpha)$ for $T_\alpha=(P_\alpha)(Q_\alpha)$, respectively.  By Theorem~\ref{theorem:m_qubit_minimal_injectivity}, injectivity is equivalent to both binary matrices having rank $4$.  Equivalently, both determinants over $\mathbb F_2$ are nonzero:
\begin{equation}
    \det_{\mathbb F_2}V_L(S)=\det_{\mathbb F_2}V_R(S)=1.
\end{equation}
Here $\det_{\mathbb F_2}$ denotes the determinant over $\mathbb F_2$.  Enumerating all four-element label sets and evaluating the two binary determinants gives
\begin{multline}
\label{eq:minor_generating_identity}
    \sum_{\substack{S\subset\{1,\ldots,13\}\\ |S|=4}}
    \det_{\mathbb F_2}V_L(S)
    \det_{\mathbb F_2}V_R(S)\times\prod_{\alpha\in S}\xi_\alpha\\
    =(\xi_6+\xi_7)(\xi_8+\xi_9)
    (\xi_{10}+\xi_{11})(\xi_{12}+\xi_{13}).
\end{multline}
Here $\xi_\alpha$ is a bookkeeping variable attached to the term $T_\alpha$ in Table~\ref{tab:fredkin_terms_binary}.  More precisely, Eq.~\eqref{eq:minor_generating_identity} is an identity in the commutative polynomial ring $\mathbb F_2[\xi_1,\ldots,\xi_{13}]$.  These variables only record which Pauli-string terms are chosen.  For a label set $S$, the monomial $\prod_{\alpha\in S}\xi_\alpha$ records the subset of terms included in $h_S^{(2)}$.  For example, the label set $S=\{6,8,10,12\}$ is represented by the monomial $\xi_6\xi_8\xi_{10}\xi_{12}$.  The coefficient of $\prod_{\alpha\in S}\xi_\alpha$ is one precisely when both determinants are nonzero.  Hence the injective models are exactly the sixteen choices in Eq.~\eqref{eq:sixteen_term_set_condition}. 
\end{proof}

Theorem~\ref{theorem:minimal_injective_fredkin_four_terms} shows that one must choose exactly one term from each of the four pairs $\{T_6,T_7\}$, $\{T_8,T_9\}$, $\{T_{10},T_{11}\}$, and $\{T_{12},T_{13}\}$.  Therefore the deformed Fredkin density contains exactly $2^4=16$ minimal injective models.  The first five terms in Table~\ref{tab:fredkin_terms_binary} never appear in these minimal sets.  They are present in the full deformed Fredkin Hamiltonian, but including any of them in a four-term density leaves either the left or the right binary span below rank $4$.  Consequently, the sixteen minimal models are insensitive to the special point $t=1$, since the terms proportional to $t^2-1$ are not used.

\subsection{Minimal models in qubit representation}
\label{subsec:minimal_models_qubit_representation}

A minimal injective Hamiltonian density in the composite spin representation does not automatically correspond to a translation-invariant three-site spin-$1/2$ Hamiltonian density.  This is because a generic composite nearest-neighbor term becomes a four-site operator in the original spin-$1/2$ chain:
\begin{equation}
    (AB)_n(CD)_{n+1}
    \longleftrightarrow
    A_{2n-1}B_{2n}C_{2n+1}D_{2n+2}.
\end{equation}
Thus a composite Hamiltonian density can be unblocked to a genuine three-site Hamiltonian density only when its terms occur in the paired form specified by the following criterion.

\begin{lemma}
\label{theorem:unblocking_to_three_site_qubit}
A two-qubit composite nearest-neighbor Hamiltonian density $h_n^{(2)}$ is the blocked form of a one-site translation-invariant spin-$1/2$ three-site density
\begin{equation}
    \widetilde h_j^{(3)}=\sum_{A,B,C}\kappa_{ABC}A_jB_{j+1}C_{j+2}
\end{equation}
if and only if
\begin{equation}
\label{eq:partner_pair_form}
    h_n^{(2)}=\sum_{A,B,C}\kappa_{ABC}
    \big((AB)_n(CI)_{n+1}+(IA)_n(BC)_{n+1}\big),
\end{equation}
where $\kappa_{ABC}$ denotes the coefficients. 
\end{lemma}

%Thus each three-site Pauli string $ABC$ gives the partner pair $(AB)(CI)$ and $(IA)(BC)$ with equal coefficients.

\begin{proof}
Starting from the three-site Hamiltonian density $A_jB_{j+1}C_{j+2}$, we separate the one-site translation-invariant sum into odd and even starting positions.  For $j=2n-1$,
\begin{equation}
    A_{2n-1}B_{2n}C_{2n+1}=(AB)_n(CI)_{n+1},
\end{equation}
whereas for $j=2n$,
\begin{equation}
    A_{2n}B_{2n+1}C_{2n+2}=(IA)_n(BC)_{n+1}.
\end{equation}
One-site translation invariance requires the odd- and even-starting contributions to have the same coefficient.  This proves necessity.  Conversely, any density of the form in Eq.~\eqref{eq:partner_pair_form} unblocks to the one-site translation-invariant three-site density above, so the condition is also sufficient.
\end{proof}

For the sixteen minimal injective models obtained from the deformed Fredkin spin chain, only four partner pairs are relevant.  Their corresponding three-site Hamiltonian densities are
\begin{subequations}
\label{eq:unblocked_partner_pairs}
\begin{align}
    (T_6,T_8)&\leftrightarrow X_jX_{j+1}Z_{j+2}, \label{eq:unblocked_partner_pair_xxz}\\
    (T_7,T_9)&\leftrightarrow Y_jY_{j+1}Z_{j+2}, \label{eq:unblocked_partner_pair_yyz}\\
    (T_{10},T_{12})&\leftrightarrow Z_jX_{j+1}X_{j+2}, \label{eq:unblocked_partner_pair_zxx}\\
    (T_{11},T_{13})&\leftrightarrow Z_jY_{j+1}Y_{j+2}. \label{eq:unblocked_partner_pair_zyy}
\end{align}
\end{subequations}
The two terms in each partner pair must have the same coefficient. 

%Combining these pairs with the minimal injectivity criterion in Theorem~\ref{theorem:m_qubit_minimal_injectivity}, we obtain four minimal models that are translational spin-1/2 three-site models, summarized in Table~\ref{tab:minimal_unblocked_models}.

A minimal injective model must choose one term from each of the four pairs in Eq.~\eqref{eq:sixteen_term_set_condition}, while Lemma~\ref{theorem:unblocking_to_three_site_qubit} requires the chosen terms to form complete partner pairs.  Hence the only possibilities are to choose either $(T_6,T_8)$ or $(T_7,T_9)$, and independently either $(T_{10},T_{12})$ or $(T_{11},T_{13})$.  These four choices are summarized in Table~\ref{tab:minimal_unblocked_models}.

\begin{table}[t]
\caption{The four translationally invariant spin-$1/2$ chains with three-site interactions obtained from the sixteen minimal models in the composite spin representation.  The notation $XXZ$, for example, denotes $X_jX_{j+1}Z_{j+2}$.}
\label{tab:minimal_unblocked_models}
\centering
\begin{tabular}{c|c|c}
\hline\hline
Model & Partner pairs of $h_n^{(2)}$ & $\widetilde h_j^{(3)}$\\
\hline
I & $(T_6,T_8),(T_{10},T_{12})$ & $c_1XXZ+c_2ZXX$\\
II & $(T_6,T_8),(T_{11},T_{13})$ & $c_1XXZ+c_2ZYY$\\
III & $(T_7,T_9),(T_{10},T_{12})$ & $c_1YYZ+c_2ZXX$\\
IV & $(T_7,T_9),(T_{11},T_{13})$ & $c_1YYZ+c_2ZYY$\\
\hline\hline
\end{tabular}
\end{table}

The equal-coefficient condition within each partner pair is essential: unequal coefficients would give a two-site-periodic qubit Hamiltonian rather than a one-site translation-invariant one.  Each one-site translation-invariant model contains precisely two three-site Pauli strings, so the minimal four-term structure in the composite representation becomes a two-term structure in the spin-$1/2$ representation. 

%Their integrability properties are not analyzed here; they are discussed in detail in Sec.~\ref{subsec:three_site_model_tests}.

\section{Integrability tests for the minimal injective models}
\label{sec:integrability_minimal_composite}

Having identified the sixteen minimal injective models, we now determine which of them are integrable or nonintegrable.  We first work in the two-qubit composite representation in Sec.~\ref{subsec:reshetikhin_hokkyo_minimal}, where the Hamiltonians are nearest-neighbor and Reshetikhin's and Hokkyo's criteria apply directly. We then return to the physical spin-$1/2$ chain in Sec.~\ref{subsec:three_site_model_tests} and obtain the classification of one-site translation-invariant three-site models.

\subsection{Reshetikhin-integrable and Hokkyo-nonintegrable models}
\label{subsec:reshetikhin_hokkyo_minimal}

We consider the sixteen minimal injective models obtained from the deformed Fredkin spin chain in Theorem~\ref{theorem:minimal_injective_fredkin_four_terms}. The Hamiltonian is denoted by
\begin{equation}
\label{eq:minimal_term_set_density_general}
    H_S=\sum_nh^{(2)}_{S,n},
\end{equation}
where the Hamiltonian density has the form
\begin{equation}
h_S^{(2)}=g_{\alpha_1}T_{\alpha_1}+g_{\alpha_2}T_{\alpha_2}+g_{\alpha_3}T_{\alpha_3}+g_{\alpha_4}T_{\alpha_4},
\end{equation}
with $\prod_{k=1}^4 g_{\alpha_k}\neq0$. The label indices are chosen from $\alpha_1\in\{6,7\}$, $\alpha_2\in\{8,9\}$, $\alpha_3\in\{10,11\}$, and $\alpha_4\in\{12,13\}$, and correspond to the operators in Table~\ref{tab:fredkin_terms_binary}. For convenience, we write the label set as $S=\{\alpha_1,\alpha_2,\alpha_3,\alpha_4\}$ and the coefficients as $\{g_{\alpha_1},g_{\alpha_2},g_{\alpha_3},g_{\alpha_4}\} = \{a,b,c,d\}$.

The integrability criterion based on the Reshetikhin condition, namely Theorem~\ref{theorem:Reshetikhin_rigorous}, asks whether there exists a three-local charge of the form
\begin{equation}
\label{eq:Q3_reshetikhin_minimal}
    Q_S^{(3)}=\sum_n\left(\iota_2(h_S^{(2)})_n+\mathbf A_{S,n}^{(2)}\right)
\end{equation}
that commutes with $H_S$. Here the map $\iota_2$ is defined in Eq.~\eqref{eq:iota_2}. Applying the Reshetikhin test to the sixteen minimal injective models gives the following theorem.

\begin{theorem}
\label{thm:reshetikhin_integrable_minimal_sets}
Among the sixteen minimal injective models, the six models listed in Table~\ref{tab:reshetikhin_minimal_term_sets} satisfy the Reshetikhin condition. For each of them, $Q_S^{(3)}$ is a conserved charge, and the corresponding composite chain is integrable.
\end{theorem}

\begin{table}[t]
\centering
\renewcommand{\arraystretch}{1.3}
\caption{The six minimal injective models satisfying the Reshetikhin condition. Here $\mathbf A_S^{(2)}$ is a two-local compensating term required for the Reshetikhin condition.}
\label{tab:reshetikhin_minimal_term_sets}
\begin{tabular}{c|c}
\hline\hline
Label index $S$ & $\mathbf A_S^{(2)}$\\
\hline
$\{6,8,10,12\}$ & $2iab(XI)(YZ)-2icd(ZY)(IX)$\\
$\{6,9,10,13\}$ & $0$\\
$\{6,9,11,12\}$ & $-2iac(IY)(ZX)+2ibd(XZ)(YI)$\\
$\{7,8,10,13\}$ & $2iac(IX)(ZY)-2ibd(YZ)(XI)$\\
$\{7,8,11,12\}$ & $0$\\
$\{7,9,11,13\}$ & $-2iab(YI)(XZ)+2icd(ZX)(IY)$\\
\hline\hline
\end{tabular}
\end{table}

\begin{proof}
The six models in Table~\ref{tab:reshetikhin_minimal_term_sets} are related pairwise by a simple one-site unitary symmetry.  Namely, a rotation about the $Z$ axis maps
\begin{equation}
    X\mapsto Y,\qquad
    Y\mapsto -X,\qquad
    Z\mapsto Z,
\end{equation}
up to irrelevant overall signs of Pauli-string operators. Under this transformation,
\begin{equation}
\begin{gathered}
    \{6,8,10,12\}\leftrightarrow\{7,9,11,13\},\\
    \{6,9,10,13\}\leftrightarrow\{7,8,11,12\},\\
    \{6,9,11,12\}\leftrightarrow\{7,8,10,13\}.
\end{gathered}
\end{equation}
It is therefore enough to verify the Reshetikhin condition for one representative of each pair.

Expand a general two-local correction term as
\begin{equation}
    \mathbf A_S^{(2)}=\sum_{P,Q\in\mathcal P_2\setminus\{II\}}q_{PQ}(P)(Q).
\end{equation}
Substituting this ansatz into $[Q_S^{(3)},H_S]=0$ and projecting onto Pauli strings gives a finite system of linear equations for the coefficients $q_{PQ}$. Row reduction gives solutions precisely for the six models in Table~\ref{tab:reshetikhin_minimal_term_sets}.

For the representative model with label set $S=\{6,8,10,12\}$, the Hamiltonian density is
\begin{multline}
h_S^{(2)}=a(IX)(XZ)+b(XX)(ZI)\\
+c(IZ)(XX)+d(ZX)(XI).
\end{multline}
The correction term $\mathbf A_S^{(2)}$ is
\begin{equation}
    \mathbf A_S^{(2)}=2iab(XI)(YZ)-2icd(ZY)(IX).
\end{equation}
The cancellation mechanism for $[Q_S^{(3)},H_S]=0$ is displayed explicitly in Appendix~\ref{App:reshetikhin_S6810}. The two remaining inequivalent representatives are checked by the same finite Pauli-string calculation, and their symmetry-related models follow from the unitary equivalence above.  The corresponding corrections are listed in Table~\ref{tab:reshetikhin_minimal_term_sets}. 
\end{proof}

The remaining ten models do not satisfy the Reshetikhin condition; this alone does not imply nonintegrability. We next test the two-local conservation condition and then apply Hokkyo's nonintegrability criterion. This identifies two of the ten models as rigorously nonintegrable. 

\begin{theorem}
\label{thm:hokkyo_nonintegrable_minimal_sets}
The minimal injective models with the index labels
\begin{equation}
\label{eq:SX_SY_def}
    S_X=\{6,8,11,13\},
    \qquad
    S_Y=\{7,9,10,12\},
\end{equation}
with the additional restriction $b\neq c$, have no nontrivial local conserved charges and are therefore nonintegrable. 
\end{theorem}

\begin{proof}
The densities of the two models are
\begin{subequations}
\label{eq:SX_SY_densities}
\begin{align}
    h_{S_X}^{(2)}
    ={}&a(IX)(XZ)+b(XX)(ZI)\notag\\
    &+c(IZ)(YY)+d(ZY)(YI), \label{eq:SX_density}\\
    h_{S_Y}^{(2)}
    ={}&a(IY)(YZ)+b(YY)(ZI)\notag\\
    &+c(IZ)(XX)+d(ZX)(XI). \label{eq:SY_density}
\end{align}
\end{subequations}
Solving the homogeneous equations that define $\mathcal B_{\le}^{(2)}(H_S)$ gives
\begin{equation}
\label{eq:B2le_SX}
    \mathcal B_{\le}^{(2)}(H_{S_X})=
    \begin{cases}
    \mathbb C H_{S_X}^{(2)}, & b\neq c,\\[1mm]
    \mathbb C H_{S_X}^{(2)}\oplus\mathbb C\sum_n(IZ)_n(ZI)_{n+1}, & b=c,
    \end{cases}
\end{equation}
with the same result for $S_Y$ after exchanging $X$ and $Y$:
\begin{equation}
\label{eq:B2le_SY}
    \mathcal B_{\le}^{(2)}(H_{S_Y})=
    \begin{cases}
    \mathbb C H_{S_Y}^{(2)}, & b\neq c,\\[1mm]
    \mathbb C H_{S_Y}^{(2)}\oplus\mathbb C\sum_n(IZ)_n(ZI)_{n+1}, & b=c.
    \end{cases}
\end{equation}
The other eight models that violate the Reshetikhin condition also fail the two-local conservation condition. They are summarized in Table~\ref{tab:extra_B2le_nonres_sets}.  In each case, one Hamiltonian term commutes with the other three terms.  Its coefficient is therefore not fixed by the cancellation of the leading three-local commutators, leaving an extra translationally invariant two-local operator not proportional to $H_S^{(2)}$.

\begin{table}[t]
\centering
\caption{The eight models that violate the Reshetikhin condition and also fail the two-local conservation condition. The second column gives the commuting term whose translational sum is an element of $\mathcal B_{\le}^{(2)}(H_S)\setminus\mathbb C H_S^{(2)}$.}
\label{tab:extra_B2le_nonres_sets}
\begin{tabular}{c|c}
\hline\hline
~Label index~ $S$ & Commuting term\\
\hline
$\{6,8,10,13\}$ & $(IZ)(XX)$\\
$\{6,8,11,12\}$ & $(ZX)(XI)$\\
$\{6,9,10,12\}$ & $(IX)(XZ)$\\
$\{6,9,11,13\}$ & $(YY)(ZI)$\\
$\{7,8,10,12\}$ & $(XX)(ZI)$\\
$\{7,8,11,13\}$ & $(IY)(YZ)$\\
$\{7,9,10,13\}$ & $(ZY)(YI)$\\
$\{7,9,11,12\}$ & $(IZ)(YY)$\\
\hline\hline
\end{tabular}
\end{table}

It remains to show that $S_X$ and $S_Y$ fail the admissible three-local test in Eq.~\eqref{eq:Hokkyo_nonintegrable_inequality}.  We use the same method as in Sec.~\ref{integrability_test_deform}: we exhibit a length-three Pauli-string operator in the commutator generated by the map $\iota_2$ that cannot be produced by any two-local correction. For $S_X$, the three-local part of $[\iota_2(h_{S_X}^{(2)}),H_{S_X}]$ contains $-2acd(IZ)(XX)(ZZ)$ and $2abd(ZZ)(YY)(ZI)$. Neither string can be produced by $[\mathbf A^{(2)},H_{S_X}]$. For instance, $(IZ)(XX)(ZZ)$ cannot arise from $[\mathbf A_n^{(2)},h^{(2)}_{S_X,n+1}]$ because no right operator in $h_{S_X}^{(2)}$ is $ZZ$. It also cannot arise from $[h^{(2)}_{S_X,n},\mathbf A_{n+1}^{(2)}]$.  In this overlap, the first composite-site operator $IZ$ forces the Hamiltonian term to be $(IZ)(YY)$.  The desired middle operator would then require a one-site commutator of $YY$ with a Pauli string to produce $XX$. Since $YY$ and $XX$ commute, this overlap cannot generate the obstruction. The same argument applies to the operator $(ZZ)(YY)(ZI)$. 

For $S_Y$, one obtains the $X\leftrightarrow Y$ counterpart, namely the operators $-2acd(IZ)(YY)(ZZ)$ and $2abd(ZZ)(XX)(ZI)$ from $[\iota_2(h_{S_Y}^{(2)}),H_{S_Y}]$, which are likewise outside the image of $\mathbf A^{(2)}\mapsto[\mathbf A^{(2)},H_{S_Y}]$. Thus Eq.~\eqref{eq:Hokkyo_nonintegrable_inequality} has no solution. For $b\neq c$, the two models have no nontrivial local conserved charges and are therefore nonintegrable.
\end{proof}

The classification of the sixteen minimal injective models based on these integrability and nonintegrability criteria is summarized in Table~\ref{table:min_H_classfication}. The eight undetermined minimal models violate the Reshetikhin condition and also fail the two-local conservation condition required by Hokkyo's criterion. Preliminary numerical searches indicate that these models possess nontrivial three-local conserved charges but do not fit the Yang-Baxter integrable structures examined here. We leave the detailed study of these special integrable models to future work. 

\begin{table}[t]
\centering
\renewcommand{\arraystretch}{1.25}
\caption{Classification of the sixteen minimal injective models defined in Eq.~\eqref{eq:h_S^2} and extracted from the deformed Fredkin spin model. The two nonintegrable models have $b\neq c$. The integrable class consists of models that satisfy the Reshetikhin condition. The nonintegrable class consists of models that satisfy Hokkyo's criterion. The undetermined class consists of models that fail both the Reshetikhin condition and the two-local conservation condition required by Hokkyo's criterion.}
\label{table:min_H_classfication}
\begin{tabular}{c|c}
\hline\hline
Class & Label index $S$ \\
\hline
Integrable &
\begin{tabular}[c]{@{}c@{}}
$\{6,8,10,12\}$, $\{6,9,10,13\}$\\
$\{6,9,11,12\}$, $\{7,8,10,13\}$\\
$\{7,8,11,12\}$, $\{7,9,11,13\}$
\end{tabular} \\
\hline
Nonintegrable &
$\{6,8,11,13\}$, $\{7,9,10,12\}$
\\
\hline
Undetermined &
\begin{tabular}[c]{@{}c@{}}
$\{6,8,10,13\}$, $\{6,8,11,12\}$\\
$\{6,9,10,12\}$, $\{6,9,11,13\}$\\
$\{7,8,10,12\}$, $\{7,8,11,13\}$\\
$\{7,9,10,13\}$, $\{7,9,11,12\}$
\end{tabular} \\
\hline\hline
\end{tabular}
\end{table}

\subsection{Integrability tests for the translation-invariant three-site models}
\label{subsec:three_site_model_tests}

Having classified the sixteen minimal injective models in the composite-spin representation, we now return to the physical spin-$1/2$ chain. As shown in Sec.~\ref{subsec:minimal_models_qubit_representation}, only four of the minimal injective models admit a spin-$1/2$ realization with three-site interactions. Their integrability can be determined rigorously, as stated in the following theorem. 

\begin{table*}[t]
\centering
\renewcommand{\arraystretch}{1.35}
\caption{Classification of the four minimal three-site spin-$1/2$ models.  Here $c_1c_2\neq0$.}
\label{tab:minimal_models_summary}
\begin{tabular}{c|c|c|c}
\hline\hline
~Model~ & ~Label index~ $S$ & Three-site density $\widetilde h_j^{(3)}$ & Class\\
\hline
$\widetilde H_{\rm I}$ & $\{6,8,10,12\}$ & $c_1X_jX_{j+1}Z_{j+2}+c_2Z_jX_{j+1}X_{j+2}$ & Integrable \\
$\widetilde H_{\rm II}$ & $\{6,8,11,13\}$ & $c_1X_jX_{j+1}Z_{j+2}+c_2Z_jY_{j+1}Y_{j+2}$ & ~Nonintegrable for $c_1\neq c_2$~ \\
$\widetilde H_{\rm III}$ & $\{7,9,10,12\}$ & $c_1Y_jY_{j+1}Z_{j+2}+c_2Z_jX_{j+1}X_{j+2}$ & ~Nonintegrable for $c_1\neq c_2$~ \\
$\widetilde H_{\rm IV}$ & $\{7,9,11,13\}$ & $c_1Y_jY_{j+1}Z_{j+2}+c_2Z_jY_{j+1}Y_{j+2}$ & Integrable \\
\hline\hline
\end{tabular}
\end{table*}

\begin{theorem}[Integrability classification of three-site models]
\label{thm:three_site_integrability}
The four models $\widetilde H_{\rm I}$, $\widetilde H_{\rm II}$, $\widetilde H_{\rm III}$, and $\widetilde H_{\rm IV}$ with the Hamiltonian densities listed in Table~\ref{tab:minimal_unblocked_models} are classified as follows: $\widetilde H_{\rm I}$ and $\widetilde H_{\rm IV}$ are integrable, while $\widetilde H_{\rm II}$ and $\widetilde H_{\rm III}$ are nonintegrable for generic couplings $c_1c_2(c_1-c_2)\neq0$.
\end{theorem}

\begin{proof}
The grouping map is a local algebra isomorphism, so it preserves commutators and maps local charges to local charges, changing their range only by a finite factor. For
\begin{equation}
    \widetilde H_{\rm I}=\sum_j \left(c_1X_jX_{j+1}Z_{j+2}+c_2Z_jX_{j+1}X_{j+2} \right),
\end{equation}
the Hamiltonian density in the composite-spin representation is
\begin{equation}
\label{eq:block_density_I}
    h_{\rm I,n}^{(2)}=c_1(T_6+T_8)+c_2(T_{10}+T_{12}).
\end{equation}
The model satisfies the Reshetikhin condition listed in Table~\ref{tab:reshetikhin_minimal_term_sets}. With $a=b=c_1$ and $c=d=c_2$, the required correction operator is
\begin{equation}
\label{eq:A2_unblocked_I}
    \mathbf A_{\rm I}^{(2)}=2ic_1^2(XI)(YZ)-2ic_2^2(ZY)(IX)
\end{equation}
for the conservation of the three-local charge. In the original spin-$1/2$ representation, the conserved charge is
\begin{equation}
    \widetilde Q_{\rm I}^{(5)}
    =\sum_j\left(
    [\widetilde h_{{\rm I},j}^{(3)},\widetilde h_{{\rm I},j+2}^{(3)}]
    +\mathbf A_j^{(4)}\right),
\end{equation}
where the correction operator becomes
\begin{equation}
    \mathbf A_j^{(4)}
    =2ic_1^2X_jI_{j+1}Y_{j+2}Z_{j+3}
    -2ic_2^2Z_jY_{j+1}I_{j+2}X_{j+3}.
\end{equation}

The same argument applies to $\widetilde H_{\rm IV}$, which has the Hamiltonian density
\begin{equation}
\label{eq:block_density_IV}
    h_{{\rm IV},n}^{(2)}=c_1(T_7+T_9)+c_2(T_{11}+T_{13}).
\end{equation}
The correction operator for the Reshetikhin condition, as listed in Table~\ref{tab:reshetikhin_minimal_term_sets}, is
\begin{equation}
\label{eq:A2_unblocked_IV}
    \mathbf A_{\rm IV}^{(2)}=-2ic_1^2(YI)(XZ)+2ic_2^2(ZX)(IY).
\end{equation}
Thus $\widetilde H_{\rm I}$ and $\widetilde H_{\rm IV}$ inherit conserved charges from the composite representation and are integrable.

For $\widetilde H_{\rm II}$ and $\widetilde H_{\rm III}$, the Hamiltonian densities are respectively
\begin{subequations}
\label{eq:block_densities_II_III}
\begin{align}
    h_{\rm II}^{(2)}
    & = c_1(T_6+T_8)+c_2(T_{11}+T_{13}), \label{eq:block_density_II}\\
    h_{\rm III}^{(2)}
    & = c_1(T_7+T_9)+c_2(T_{10}+T_{12}). \label{eq:block_density_III}
\end{align}
\end{subequations}
These are exactly the models $S_X$ and $S_Y$ defined in Eq.~\eqref{eq:SX_SY_def}. For $c_1\neq c_2$, Theorem~\ref{thm:hokkyo_nonintegrable_minimal_sets} establishes the two-local conservation condition and the absence of an admissible three-local charge.  Hence there are no composite conserved charges of length $3\le k\le L/2$.

It remains to exclude charges that become at most two-local after grouping, since such charges could still be nontrivial in the original spin-$1/2$ representation. Let $Q$ be such a conserved operator for either $\widetilde H_{\rm II}$ or $\widetilde H_{\rm III}$. Its strictly two-local part must lie in $\mathcal B_{\le}^{(2)}$, hence is proportional to the Hamiltonian itself.  Subtracting the corresponding multiple of the Hamiltonian leaves a possible one-site conserved operator
\begin{equation}
    R=\sum_nR_n,
    \qquad
    R_n\in\operatorname{span}\{P_n\;|\;P\in\mathcal P_2,\ P\neq II\}.
\end{equation}
The conservation equation is
\begin{equation}
\label{eq:one_local_conservation_unblocked}
    [R_n+R_{n+1},h_{\mu,n}^{(2)}]=0.
\end{equation}
Expanding $R_n$ in the two-qubit Pauli-string basis gives a homogeneous system of linear equations. For both models $\widetilde H_{\rm II}$ and $\widetilde H_{\rm III}$, this system has full rank $15$ whenever $c_1c_2(c_1-c_2)\neq0$, so $R=0$.  Hence no nontrivial lower-range conserved charge survives. The models are therefore nonintegrable for $c_1\neq c_2$ and provide minimal nonintegrable models with three-site interactions obtained from the Fredkin spin chain studied here. The results are summarized in Table~\ref{tab:minimal_models_summary}.
\end{proof}

\section{Frustration-graph analysis of the minimal models}
\label{Sec:test_frustration_graph}

The preceding charge-based analysis classified the minimal injective models into integrable and nonintegrable classes. We now give a complementary frustration-graph interpretation of this classification. In Sec.~\ref{subsec:frustration_graph_composite}, we analyze the models in the composite-spin representation and show that the integrable classes decompose into FFD-type graph components, whereas the nonintegrable models contain induced claws. In Sec.~\ref{subsec:frustration_graph_three_site}, we repeat the analysis for the physical spin-$1/2$ chain and identify the corresponding graph structures for the four three-site models.

\subsection{Frustration graphs of composite-spin models}
\label{subsec:frustration_graph_composite}

A nearest-neighbor Hamiltonian in the composite-spin representation is built from Pauli-string terms $T_{\alpha,n}=(P_\alpha)_n(Q_\alpha)_{n+1}$, where $(P_\alpha,Q_\alpha)$ are listed in Table~\ref{tab:fredkin_terms_binary}. Two vertices anticommute only when the corresponding strings overlap on one or two composite sites.  In terms of the binary symplectic product, the two Pauli strings $T_{\alpha,n}$ and $T_{\beta,n}$ are connected if
\begin{equation}
\langle P_\alpha,P_\beta\rangle_{\rm sp}
    +\langle Q_\alpha,Q_\beta\rangle_{\rm sp}=1.
\end{equation}
Similarly, for the Pauli strings $T_{\alpha,n}$ and $T_{\beta,n+1}$, the connectivity condition is
\begin{equation}
    \langle Q_\alpha,P_\beta\rangle_{\rm sp}=1. 
\end{equation}

We first focus on the six integrable models listed in Table~\ref{tab:reshetikhin_minimal_term_sets}, which satisfy the Reshetikhin condition. Their frustration graphs split into two classes:
\begin{subequations}
\begin{equation}
\label{eq:C_A}
\mathcal C_{\rm A}=\{\{6,9,10,13\},\{7,8,11,12\}\},
\end{equation}
\begin{multline}
\label{eq:C_B}
\mathcal C_{\rm B}=\{\{6,8,10,12\},\{6,9,11,12\}, \\
 \{7,8,10,13\},\{7,9,11,13\}\}.
\end{multline}
\end{subequations}
For class $\mathcal C_{\rm A}$, the graph is a disjoint union of four identical chains, one for each term in the Hamiltonian density. Each chain has edges $T_{\alpha,n}\sim T_{\alpha,n+1}$, where the symbol $\sim$ denotes connectivity, as shown in Fig.~\ref{fig:fg-composite-tableIII}(a).  

For class $\mathcal C_{\rm B}$, the graph is a disjoint union of two identical triangulated ladders. It therefore suffices to display one representative ladder, generated by two terms \(T_{\alpha_1}\) and \(T_{\alpha_2}\), whose local edges are $T_{\alpha_1,n}\sim T_{\alpha_1,n+1}$, $T_{\alpha_2,n}\sim T_{\alpha_2,n+1}$, $T_{\alpha_1,n}\sim T_{\alpha_2,n}$, and $T_{\alpha_1,n}\sim T_{\alpha_2,n+1}$. The representative ladder is shown in Fig.~\ref{fig:fg-composite-tableIII}(b). Both graph types are locally claw-free and even-hole-free. The frustration-graph criterion therefore agrees with the Reshetikhin condition and identifies these models as FFD-type integrable models.

The two nonintegrable models are $S_X=\{6,8,11,13\}$ and $S_Y=\{7,9,10,12\}$. They are related by $X\leftrightarrow Y$ and have isomorphic frustration graphs. For $S_X$, the local edges are
\begin{equation}
\begin{array}{@{}l@{\qquad}l@{}}
T_{\alpha_1,n}\sim T_{\alpha_2,n}, &
T_{\alpha_1,n}\sim T_{\alpha_3,n}, \\
T_{\alpha_2,n}\sim T_{\alpha_4,n}, &
T_{\alpha_3,n}\sim T_{\alpha_4,n}, \\
T_{\alpha_1,n}\sim T_{\alpha_1,n+1}, &
T_{\alpha_1,n}\sim T_{\alpha_2,n+1}, \\
T_{\alpha_2,n}\sim T_{\alpha_2,n+1}, &
T_{\alpha_3,n}\sim T_{\alpha_1,n+1}, \\
T_{\alpha_3,n}\sim T_{\alpha_3,n+1}, &
T_{\alpha_3,n}\sim T_{\alpha_4,n+1}, \\
T_{\alpha_4,n}\sim T_{\alpha_2,n+1}, &
T_{\alpha_4,n}\sim T_{\alpha_4,n+1}.
\end{array}
\end{equation}
The corresponding frustration graph is shown in Fig.~\ref{fig:fg-composite-tableIII}(c).  The induced subgraph with center $T_{\alpha_3,n}$ and leaves $T_{\alpha_1,n}$, $T_{\alpha_4,n}$, and $T_{\alpha_3,n+1}$ is a claw. Thus these two models are not covered by the FFD integrability criterion. For $b\neq c$, this result is consistent with the charge-based proof that they are genuinely nonintegrable.

\begin{figure}[t]
\centering
\begin{tikzpicture}[x=1.05cm,y=0.675cm,
  v/.style={circle,fill=black,inner sep=1.1pt},
  e/.style={line width=0.35pt},
  lab/.style={font={\normalfont\normalsize},fill=white,inner sep=0.45pt},
  panel/.style={font={\normalfont\normalsize},inner sep=0pt}]
% Upper: one chain for class C_A.
\node[panel,anchor=east] at (-1.75,0) {(a)};
\node[lab,anchor=east] at (-0.45,0) {$\mathcal C_{\rm A}$};
\foreach \i in {0,1,2,3,4}{
  \node[v] (a\i) at ({1.0*\i},0) {};
}
\foreach \i/\j in {0/1,1/2,2/3,3/4}{
  \draw[e] (a\i)--(a\j);
}

% Middle: one triangulated ladder for class C_B.
\begin{scope}[yshift=-1.875cm]
\node[panel,anchor=east] at (-1.75,0.15) {(b)};
\node[lab,anchor=east] at (-0.45,0.15) {$\mathcal C_{\rm B}$};
\foreach \i in {0,1,2,3,4}{
  \node[v] (b1\i) at ({1.0*\i},0.65) {};
  \node[v] (b2\i) at ({1.0*\i},-0.35) {};
}
\foreach \i/\j in {0/1,1/2,2/3,3/4}{
  \draw[e] (b1\i)--(b1\j);
  \draw[e] (b2\i)--(b2\j);
  \draw[e] (b1\i)--(b2\j);
}
\foreach \i in {0,1,2,3,4}{
  \draw[e] (b1\i)--(b2\i);
}
\end{scope}

% Lower: the Hokkyo-nonintegrable claw graph.
\begin{scope}[yshift=-4.20cm,
  v/.style={circle,fill=black,inner sep=1.05pt},
  hnode/.style={circle,draw=red,fill=white,inner sep=1.35pt,line width=0.55pt},
  e/.style={line width=0.30pt,black!65},
  he/.style={line width=0.85pt,red}]
\node[panel,anchor=east] at (-1.75,0.8) {(c)};
\node[lab,anchor=east] at (-0.45,0) {$\begin{gathered}S_X\\ S_Y\end{gathered}$};
\foreach \i in {0,1,2}{
  \node[v] (c1\i) at ({1.45*\i+0.32},1.05) {};
  \node[v] (c2\i) at ({1.45*\i+0.64},0.35) {};
  \node[v] (c3\i) at ({1.45*\i},-0.35) {};
  \node[v] (c4\i) at ({1.45*\i+0.32},-1.05) {};
  \draw[e] (c1\i)--(c2\i)--(c4\i)--(c3\i)--(c1\i);
}
\foreach \i/\j in {0/1,1/2}{
  \draw[e] (c1\i)--(c1\j);
  \draw[e] (c2\i)--(c2\j);
  \draw[e] (c3\i)--(c3\j);
  \draw[e] (c4\i)--(c4\j);
  \draw[e] (c1\i)--(c2\j);
  \draw[e] (c3\i)--(c1\j);
  \draw[e] (c3\i)--(c4\j);
  \draw[e] (c4\i)--(c2\j);
}
\draw[he] (c31)--(c11);
\draw[he] (c31)--(c41);
\draw[he] (c31)--(c32);
\foreach \p in {c11,c41,c31,c32}{\node[hnode] at (\p) {};}
\end{scope}
\end{tikzpicture}
\caption{Frustration graphs of the minimal Hamiltonians in the composite two-qubit representation. Only finite windows are shown.  (a) The class $\mathcal C_{\rm A}$ defined in Eq.~\eqref{eq:C_A}. (b) The class $\mathcal C_{\rm B}$ defined in Eq.~\eqref{eq:C_B}.  (c) The graphs for the Hamiltonian densities $S_X$ and $S_Y$ defined in \eqref{eq:SX_SY_def}. In panel (c), the red edges and vertices exhibit an induced claw.}
\label{fig:fg-composite-tableIII}
\end{figure}

\subsection{Frustration graphs of the three-site models}
\label{subsec:frustration_graph_three_site}

We next focus on the four three-site models listed in Table~\ref{tab:minimal_models_summary}. To keep the frustration-graph notation compact, denote the two local terms in each Hamiltonian density by
\begin{equation}
\label{eq:tableIV-pq-definition}
\begin{array}{ll}
\widetilde H_{\rm I}:&(p_j,q_j)=(X_jX_{j+1}Z_{j+2},Z_jX_{j+1}X_{j+2}),\\
\widetilde H_{\rm II}:&(p_j,q_j)=(X_jX_{j+1}Z_{j+2},Z_jY_{j+1}Y_{j+2}),\\
\widetilde H_{\rm III}:&(p_j,q_j)=(Y_jY_{j+1}Z_{j+2},Z_jX_{j+1}X_{j+2}),\\
\widetilde H_{\rm IV}:&(p_j,q_j)=(Y_jY_{j+1}Z_{j+2},Z_jY_{j+1}Y_{j+2}),
\end{array}
\end{equation}
with $\widetilde H_\mu=\sum_j(c_1p_j+c_2q_j)$.  Each isolated $p$ or $q$ family forms an FFD chain, with local frustration-graph edges $p_j\sim p_{j+1}$, $p_j\sim p_{j+2}$, $q_j\sim q_{j+1}$, and $q_j\sim q_{j+2}$. 

For $\widetilde H_{\rm I}$ and $\widetilde H_{\rm IV}$, the two operator families commute with one another:
\begin{equation}
\label{eq:tableIV-decoupled}
    [p_j,q_{j'}]=0,
    \qquad \forall j,j'.
\end{equation}
Thus the graph is a disjoint union of two identical FFD components.  One representative component is shown in Fig.~\ref{fig:fg-tableIV}(a).  This agrees with the conserved-charge analysis, which identifies these two models as integrable.

For the remaining two models, $\widetilde H_{\rm II}$ and $\widetilde H_{\rm III}$, the $p$ and $q$ families are coupled by additional mixed anticommutation edges $p_j\sim q_j$ and $p_j\sim q_{j-2}$. The resulting graph is shown in Fig.~\ref{fig:fg-tableIV}(b).  The induced subgraph on $p_j$, $p_{j-2}$, $p_{j+1}$, and $q_j$ is a claw. Therefore these two models do not satisfy the FFD frustration-graph criterion. For $c_1\neq c_2$, the charge-based analysis above establishes that they are genuinely nonintegrable. As above, these are local open-chain graph representations.  Closing a finite periodic chain may introduce additional global cycles, but these do not change the local claw structure relevant here.

\begin{figure}[t]
\centering
\begin{tikzpicture}[x=0.625cm,y=0.55cm,
  v/.style={circle,fill=black,inner sep=1.05pt},
  e/.style={line width=0.35pt},
  lab/.style={font={\normalfont\normalsize}},
  panel/.style={font={\normalfont\normalsize},inner sep=0pt}]
\node[panel,anchor=west] at (-1.45,1.50) {(a)};
\foreach \i in {0,1,2}{
  \node[v] (fo\i) at ({1.20*\i},0.45) {};
  \node[v] (fe\i) at ({1.20*\i+0.60},-0.40) {};
}
\foreach \i/\j in {0/1,1/2}{
  \draw[e] (fo\i)--(fo\j);
  \draw[e] (fe\i)--(fe\j);
  \draw[e] (fe\i)--(fo\j);
}
\foreach \i in {0,1,2}{\draw[e] (fo\i)--(fe\i);}
\node[lab,anchor=east] at (-0.35,0.45) {$p_{2m+1}$};
\node[lab,anchor=east] at (-0.35,-0.40) {$p_{2m}$};
\node[lab] at (1.50,-1.10) {$\widetilde H_{\rm I},\ \widetilde H_{\rm IV}$};

\begin{scope}[xshift=3.50cm,yshift=-0.69cm,
  v/.style={circle,fill=black,inner sep=1.05pt},
  hnode/.style={circle,draw=red,fill=white,inner sep=1.45pt,line width=0.55pt},
  e/.style={line width=0.32pt,black!70},
  mixv/.style={line width=0.45pt,black!55},
  mixd/.style={line width=0.45pt,black!55,dashed},
  he/.style={line width=0.90pt,red},
  lab/.style={font={\normalfont\normalsize}}]
\node[panel,anchor=west] at (-1.45,2.75) {(b)};
\foreach \i in {0,1,2}{
  \node[v] (po\i) at ({1.20*\i},2.30) {};
  \node[v] (pe\i) at ({1.20*\i+0.60},1.45) {};
}
\foreach \i/\j in {0/1,1/2}{
  \draw[e] (po\i)--(po\j);
  \draw[e] (pe\i)--(pe\j);
  \draw[e] (pe\i)--(po\j);
}
\foreach \i in {0,1,2}{\draw[e] (po\i)--(pe\i);}
\foreach \i in {0,1,2}{
  \node[v] (qo\i) at ({1.20*\i},-0.15) {};
  \node[v] (qe\i) at ({1.20*\i+0.60},-1.00) {};
}
\foreach \i/\j in {0/1,1/2}{
  \draw[e] (qo\i)--(qo\j);
  \draw[e] (qe\i)--(qe\j);
  \draw[e] (qe\i)--(qo\j);
}
\foreach \i in {0,1,2}{\draw[e] (qo\i)--(qe\i);}
\foreach \i in {0,1,2}{
  \draw[mixv] (po\i)--(qo\i);
  \draw[mixv] (pe\i)--(qe\i);
}
\foreach \i/\j in {0/1,1/2}{
  \draw[mixd] (qo\i)--(po\j);
  \draw[mixd] (qe\i)--(pe\j);
}
\draw[he] (po1)--(po0);
\draw[he] (po1)--(pe1);
\draw[he] (po1)--(qo1);
\foreach \p in {po1,po0,pe1,qo1}{\node[hnode] at (\p) {};}
\node[lab,anchor=south] at (po1) {$p_3$};
\node[lab,anchor=south] at (po0) {$p_1$};
\node[lab,anchor=north west] at (pe1) {$p_4$};
\node[lab,anchor=north] at (qo1) {$q_3$};
\node[lab,anchor=west] at (3.40,2.30) {$p_{2m+1}$};
\node[lab,anchor=west] at (4.00,1.45) {$p_{2m}$};
\node[lab,anchor=west] at (3.40,-0.15) {$q_{2m+1}$};
\node[lab,anchor=west] at (4.00,-1.00) {$q_{2m}$};
\node[lab] at (1.50,-1.85) {$\widetilde H_{\rm II},\ \widetilde H_{\rm III}$};
\end{scope}
\end{tikzpicture}
\caption{Frustration graphs of the four three-site spin-$1/2$ models listed in Table~\ref{tab:minimal_models_summary}.  (a) One representative FFD structure for the integrable models $\widetilde H_{\rm I}$ and $\widetilde H_{\rm IV}$; the second structure is identical and disconnected.  (b) The graph of the models $\widetilde H_{\rm II}$ and $\widetilde H_{\rm III}$, where mixed edges generate an induced claw.}
\label{fig:fg-tableIV}
\end{figure}

\section{Conclusions and outlook}
\label{sec:conclusion}

We have studied how integrability breaks down in translationally invariant spin-$1/2$ chains with genuine three-site interactions.  The starting point was the deformed Fredkin chain, whose Hamiltonian is composed of several integrable building blocks, including XXZ-type, Ising-type, and FFD-type terms. By grouping two physical spins into one composite site, we mapped the model to a nearest-neighbor chain with local dimension four and applied Hokkyo's nonintegrability criterion. We verified injectivity and the two-local conservation condition, and then found an explicit three-local obstruction that cannot be cancelled by any two-local correction. This proves, as stated in Theorem~\ref{theorem:nonintegrable_deformed_Fredkin}, that the deformed Fredkin chain admits no infinite tower of local conserved charges for any nonzero deformation parameter, extending the known nonintegrability of the original Fredkin chain to the deformed family.

We then derived a necessary and sufficient condition for minimal injective models, proved in Theorem~\ref{theorem:m_qubit_minimal_injectivity}. Using this criterion, we identified the minimal injective models contained in the composite Fredkin Hamiltonian. For a two-qubit composite site, injectivity requires four Pauli-string terms, and the Fredkin density contains exactly sixteen such minimal models. Among them, six satisfy the Reshetikhin condition and are integrable (Theorem~\ref{thm:reshetikhin_integrable_minimal_sets}), while two satisfy the assumptions of Hokkyo's criterion and are nonintegrable for generic coefficients (Theorem~\ref{thm:hokkyo_nonintegrable_minimal_sets}). The remaining eight fail the two-local conservation condition and are therefore not classified by this direct criterion. Imposing the additional unblocking condition for one-site translation-invariant spin-$1/2$ chains with genuine three-site interactions leaves only four models. Two of them, $\widetilde H_{\rm I}$ and $\widetilde H_{\rm IV}$, are integrable and decompose into two commuting FFD components. The remaining two models, $\widetilde H_{\rm II}$ and $\widetilde H_{\rm III}$, are nonintegrable for $c_1c_2(c_1-c_2)\neq0$ and constitute the minimal nonintegrable three-site spin-$1/2$ Hamiltonians identified within this Fredkin-derived class.

The frustration-graph analysis gives a complementary picture of the same classification. The integrable composite models have claw-free and even-hole-free local graphs, corresponding to decoupled chains or triangulated ladders of FFD type. By contrast, the nonintegrable models contain induced claws, both in the composite representation and the spin-$1/2$ three-site representation. Although the frustration-graph criterion is only sufficient, this agreement shows that the charge-based obstruction has a simple local graph-theoretic signature in the present family.

Several questions remain open. First, the symmetric point $c_1=c_2$ of the two three-site models $\widetilde H_{\rm II}$ and $\widetilde H_{\rm III}$ is not resolved by the present proof, because the two-local conservation condition fails there. Preliminary numerical tests indicate the existence of nontrivial local conserved charges, although these charges do not appear to be of Yang-Baxter type. A similar classification problem remains open for the eight models listed in Table~\ref{table:min_H_classfication}. Second, the present classification is Fredkin-derived: it identifies all minimal injective models inside the deformed Fredkin model, but it is not yet a complete classification of all three-site Pauli Hamiltonians.  Extending the binary-symplectic analysis to arbitrary three-site Pauli strings would give a broader map of the boundary between integrability and nonintegrability. Finally, it would be useful to connect these charge-based results more directly to dynamical diagnostics such as ETH behavior, operator spreading, and transport. Such a connection would clarify how the integrability breaking found here manifests itself in nonequilibrium physics.

\begin{acknowledgments}
The authors thank Balázs Pozsgay for helpful discussions. This work was funded by the NSFC (Grant Nos. 12305028, 12275215, 12275214, and 12247103), and the Youth Innovation Team of Shaanxi Universities. KZ is funded by the China Postdoctoral Science Foundation under Grant No. 2025M773421, Shaanxi Province Postdoctoral Science Foundation under Grant No. 2025BSHYDZZ017, and Scientific Research Program Funded by Education Department of Shaanxi Provincial Government (Program No. 24JP186). VK is funded by the U.S. Department of Energy, Office of Science, National Quantum Information Science Research Centers, Co-Design Center for Quantum Advantage (C2QA) under Contract No. DE-SC0012704.
\end{acknowledgments}

\appendix

\section{Technical details of Hokkyo's criterion}
\label{App:Hokkyo_sets}

This appendix summarizes the technical structures behind Hokkyo's criterion \cite{hokkyo2025rigorous}. Recall the auxiliary space $\mathcal B_{\le}^{(k)}$ defined in Eq.~\eqref{eq:B_le^k}.  An element of $\mathcal B_{\le}^{(k)}$ is a $k$-local operator whose commutator with the Hamiltonian has length at most $k$, meaning that all potentially generated $(k+1)$-local leading terms cancel.

Under the injectivity condition \eqref{eq:injectivity}, there is a natural map from $\mathcal B_{\le}^{(k-1)}$ to $\mathcal B_{\le}^{(k)}$,
\begin{equation}
\label{def:iota_map}
    \iota_{k-1}:\mathcal B_{\le}^{(k-1)}\longrightarrow\mathcal B_{\le}^{(k)},
\end{equation}
whose local action is
\begin{equation}
    \iota_{k-1}(\mathbf A_j^{(k-1)})=\left[\mathbf A_j^{(k-1)},h_{j+k-2}^{(2)}\right],
\end{equation}
where $\mathbf A_j^{(k-1)}\in \mathcal B_{\le}^{(k-1)}$.
This map captures the leading part of possible higher-range local conserved charges. 

When the two-local conservation condition $\mathcal B_{\le}^{(2)}=\mathbb C H^{(2)}$ holds, the three-local leading part of any nontrivial conserved charge is fixed, up to normalization, by the translational sum of the local density $\iota_2(h^{(2)})$. To decide whether this leading part can be completed to a conserved charge, one introduces
\begin{multline}
\label{B_lt_k}
    \mathcal B_{<}^{(k)}
    =\bigl\{\mathbf A^{(k)}+\widetilde{\mathbf A}^{(k-1)}\mid
    \mathbf A^{(k)}\in\mathcal B_{\le}^{(k)},\\
    \widetilde{\mathbf A}^{(k-1)}\in\mathcal B^{(k-1)},\ 
    \mathrm{len}([\mathbf A^{(k)}\!+\!\widetilde{\mathbf A}^{(k-1)},H])\le k-1\bigr\}.
\end{multline}
Thus $\mathcal B_{<}^{(k)}$ consists of $k$-local leading pieces that can be supplemented by a lower-range correction so that the remaining commutator has length at most $k-1$. Hokkyo's theorem states that, under injectivity and the two-local conservation condition, failure of this completion already at $k=3$ propagates to all higher $k$: if no two-local correction can reduce the commutator of the translational sum of $\iota_2(h^{(2)})$ to length at most two, then there are no nontrivial conserved charges of any range $3\le k\le N/2$.

%This is the mechanism used repeatedly in the main text.  For the deformed Fredkin chain and for the two minimal models, we first verify Eq.~\eqref{eq:2-local_condition}; then we identify an explicit three-local Pauli string in $[\iota_2(h^{(2)}),H]$ that cannot be produced by $[\mathbf A^{(2)},H]$ for any two-local $\mathbf A^{(2)}$.  That single obstruction is enough to rule out an entire tower of local conserved quantities.

\section{Analytic check of the Reshetikhin condition}
\label{App:reshetikhin_S6810}

We give a direct analytic verification of the Reshetikhin condition for the representative index set $S=\{6,8,10,12\}$. The remaining index sets in Table~\ref{tab:reshetikhin_minimal_term_sets} can be checked in the same way. We keep the composite-site notation of the main text and write
\begin{multline}
\label{App:eq:S6810_density}
    h_{S,n}^{(2)}=a(IX)_n(XZ)_{n+1}
    +b(XX)_n(ZI)_{n+1}\\
    +c(IZ)_n(XX)_{n+1}
    +d(ZX)_n(XI)_{n+1},
\end{multline}
with $abcd\neq0$. The two-local correction is
\begin{equation}
\label{App:eq:S6810_correction}
    \mathbf A_{S,n}^{(2)}
    =2iab(XI)_n(YZ)_{n+1}
    -2icd(ZY)_n(IX)_{n+1}.
\end{equation}
We show that the following three-local charge candidate,
\begin{equation}
\label{App:eq:S6810_charge}
    Q_S^{(3)}
    =\sum_n\left([h_{S,n}^{(2)},h_{S,n+1}^{(2)}]
    +\mathbf A_{S,n}^{(2)}\right)
\end{equation}
commutes with $H_S=\sum_nh_{S,n}^{(2)}$.  As in the main text, each elementary Pauli commutator carries a common factor $2i$.  To display the exact operator cancellation, we keep this overall factor below.

The local density $\iota_2(h_S^{(2)})_n$ is
\begin{multline}
\label{App:eq:S6810_iota}
    \iota_2(h_S^{(2)})_n=2ia^2(IX)_n(XY)_{n+1}(XZ)_{n+2}\\
    +2iab(IX)_n(IY)_{n+1}(ZI)_{n+2}\\
    +2ib^2(XX)_n(YX)_{n+1}(ZI)_{n+2}\\
    -2ic^2(IZ)_n(XY)_{n+1}(XX)_{n+2}\\
    -2icd(IZ)_n(YI)_{n+1}(XI)_{n+2}\\
    -2id^2(ZX)_n(YX)_{n+1}(XI)_{n+2}.
\end{multline}
We first compare the strictly three-local terms in
$[\sum_n\iota_2(h_S^{(2)})_n,H_S]$ with those generated by
$[\sum_n\mathbf A_{S,n}^{(2)},H_S]$.  The nontrivial mixed contributions from the $\iota_2$ part are
\begin{subequations}
\label{App:eq:S6810_three_boost}
\begin{equation}
\begin{array}{cccc}
2ia^2&(IX) & (XY) & (XZ)\\
b&(XX) & (ZI) & {}\\ \hline
-2ia^2b&(XI) & (YY) & (XZ)
\end{array}
\end{equation}
\begin{equation}
\begin{array}{cccc}
2ib^2&(XX) & (YX) & (ZI)\\
a&{} & (IX) & (XZ)\\ \hline
2iab^2&(XX) & (YI) & (YZ)
\end{array}
\end{equation}
\begin{equation}
\begin{array}{cccc}
-2ic^2&(IZ) & (XY) & (XX)\\
d&(ZX) & (XI) & {}\\ \hline
-2ic^2d&(ZY) & (IY) & (XX)
\end{array}
\end{equation}
\begin{equation}
\begin{array}{cccc}
-2id^2&(ZX) & (YX) & (XI)\\
c&{} & (IZ) & (XX)\\ \hline
2icd^2&(ZX) & (YY) & (IX).
\end{array}
\end{equation}
\end{subequations}
The corresponding strictly three-local terms from the correction $\mathbf A_{S,n}^{(2)}$ are
\begin{subequations}
\label{App:eq:S6810_three_correction}
\begin{equation}
\begin{array}{cccc}
2iab&(XI) & (YZ) & {}\\
a&{} & (IX) & (XZ)\\ \hline
2ia^2b&(XI) & (YY) & (XZ)
\end{array}
\end{equation}
\begin{equation}
\begin{array}{cccc}
2iab&{} & (XI) & (YZ)\\
b&(XX) & (ZI) & {}\\ \hline
-2iab^2&(XX) & (YI) & (YZ)
\end{array}
\end{equation}
\begin{equation}
\begin{array}{cccc}
-2icd&(ZY) & (IX) & {}\\
c&{} & (IZ) & (XX)\\ \hline
2ic^2d&(ZY) & (IY) & (XX)
\end{array}
\end{equation}
\begin{equation}
\begin{array}{cccc}
-2icd&{} & (ZY) & (IX)\\
d&(ZX) & (XI) & {}\\ \hline
-2icd^2&(ZX) & (YY) & (IX).
\end{array}
\end{equation}
\end{subequations}
Equations~\eqref{App:eq:S6810_three_boost} and \eqref{App:eq:S6810_three_correction} cancel term by term. The remaining length-three contributions are of the free-fermion-in-disguise type and reduce, after the translational sum, to shorter-range operators.

It remains to cancel the two-local terms. From the $\iota_2$ part one obtains
\begin{subequations}
\label{App:eq:S6810_two_boost}
\begin{equation}
\begin{array}{cccc}
2iab&(IX) & (IY) & (ZI)\\
a&(IX) & (XZ) & {}\\ \hline
2ia^2b&{} & (XX) & (ZI)
\end{array}
\end{equation}
\begin{equation}
\begin{array}{cccc}
2iab&(IX) & (IY) & (ZI)\\
b&{} & (XX) & (ZI)\\ \hline
-2iab^2&(IX) & (XZ) & {}
\end{array}
\end{equation}
\begin{equation}
\begin{array}{cccc}
-2icd&(IZ) & (YI) & (XI)\\
c&(IZ) & (XX) & {}\\ \hline
2ic^2d&{} & (ZX) & (XI)
\end{array}
\end{equation}
\begin{equation}
\begin{array}{cccc}
-2icd&(IZ) & (YI) & (XI)\\
d&{} & (ZX) & (XI)\\ \hline
-2icd^2&(IZ) & (XX) & {}
\end{array}
\end{equation}
\end{subequations}
The identity entries displayed as empty slots indicate that the result is a two-site density. The two-local correction $\mathbf A_{S,n}^{(2)}$ gives the opposite densities:
\begin{subequations}
\label{App:eq:S6810_two_correction}
\begin{equation}
\begin{array}{ccc}
2iab&(XI) & (YZ)\\
a&(IX) & (XZ)\\ \hline
-2ia^2b&(XX) & (ZI)
\end{array}
\end{equation}
\begin{equation}
\begin{array}{ccc}
2iab&(XI) & (YZ)\\
b&(XX) & (ZI)\\ \hline
2iab^2&(IX) & (XZ)
\end{array}
\end{equation}
\begin{equation}
\begin{array}{ccc}
-2icd&(ZY) & (IX)\\
c&(IZ) & (XX)\\ \hline
-2ic^2d&(ZX) & (XI)
\end{array}
\end{equation}
\begin{equation}
\begin{array}{ccc}
-2icd&(ZY) & (IX)\\
d&(ZX) & (XI)\\ \hline
2icd^2&(IZ) & (XX).
\end{array}
\end{equation}
\end{subequations}
After relabeling the summation index, Eq.~\eqref{App:eq:S6810_two_correction} cancels Eq.~\eqref{App:eq:S6810_two_boost} term by term. Combining this two-local cancellation with the strictly three-local cancellation above gives
\begin{equation}
    [Q_S^{(3)},H_S]=0.
\end{equation}
Thus the model with index set $S=\{6,8,10,12\}$ satisfies the Reshetikhin condition for arbitrary nonzero coefficients $a,b,c,d$.

% Bibliography
%\bibliography{ref.bib}


\end{document}